\documentclass{article}
\usepackage{amsfonts, amsmath, amsthm, latexsym, amscd, graphicx, xypic, psfrag}
\usepackage{epic, eepic, epsfig, setspace, ifthen, amssymb, xy, yhmath, lscape}
\usepackage{bbold, amscd, paralist, graphics}
\usepackage{scalerel,stackengine}

\usepackage[colorlinks=true]{hyperref}

\def\br#1{\left(#1\right)}
\def\Br#1{\left[#1\right]}

\def\<#1,#2>{\langle #1,#2 \rangle}

\def\bothID{\rlap{\hbox to.97\wd0{\hss\vrule height.06\ht0 width.82\wd0}}
\copy0\rlap{\kern-.36\wd0\vrule height1.05\ht0 width.05\ht0}\kern.14\wd0}

\newtheorem{theorem}{Theorem}
\newtheorem{defi}[theorem]{Definition}

\newtheorem{corollary}[theorem]{Corollary}

\newtheorem{ex}{Example}[section]

\newtheorem{proposition}[theorem]{Proposition}
\newtheorem{remark}[theorem]{Remark}

\stackMath
\newcommand\reallywidehat[1]{%
\savestack{\tmpbox}{\stretchto{%
  \scaleto{%
    \scalerel*[\widthof{\ensuremath{#1}}]{\kern-.6pt\bigwedge\kern-.6pt}%
    {\rule[-\textheight/2]{1ex}{\textheight}}
  }{\textheight}%
}{0.5ex}}%
\stackon[1pt]{#1}{\tmpbox}%
}
\def\br#1{\left(#1\right)}
\def\Br#1{\left[#1\right]}

 \DeclareMathOperator{\vari}{Var}

 \DeclareMathOperator{\Cov}{Cov}

\DeclareMathOperator{\VCM}{VCM}

\begin{document}

\title{Can You Hear the Shape of a Market? Geometric Arbitrage and Spectral Theory}

\author{Simone Farinelli\\
        Core Dynamics GmbH\\
        Scheuchzerstrasse 43\\
        CH-8006 Zurich\\
        Email: simone@coredynamics.ch\\and\\
        Hideyuki Takada\\
        Department of Information Science\\
        Narashino Campus, Toho University\\
        2-2-1-Miyama, Funabashi-Shi\\ J-274-8510 Chiba\\
        Email: hideyuki.takada@is.sci.toho-u.ac.jp
        }

\maketitle

\begin{abstract}
Utilizing gauge symmetries, the Geometric Arbitrage Theory reformulates any asset model, allowing for arbitrage by means of a stochastic principal fibre bundle with a connection  whose curvature measures the ``instantaneous arbitrage capability''. The cash flow bundle is the associated vector bundle. The zero eigenspace of its connection Laplacian parameterizes all risk-neutral measures equivalent to the statistical one. A market satisfies the No-Free-Lunch-with-Vanishing-Risk (NFLVR) condition if and only if $0$ is in the discrete spectrum of the Laplacian. The Jarrow--Protter--Shimbo theory of asset bubbles and their classification and decomposition extend to markets not satisfying the NFLVR. Euler's characteristic of the asset nominal space and  non-vanishing of the homology group of the cash flow bundle are both topological obstructions to NFLVR.
\textit{AMS:} 91G80 (primary);60G46,58C40 (secondary).
\end{abstract}
\tableofcontents

\section{Introduction}
This paper develops a conceptual structure---called Geometric Arbitrage Theory (GAT)---to
link arbitrage modeling in generic markets with spectral~theory.\par
GAT rephrases classical stochastic finance in stochastic differential geometric
terms in order to characterize arbitrage.
 The main idea of the GAT approach
consists of modeling markets made of basic financial instruments
together with their term structures as principal fibre bundles.
The financial features of this market---such as no-arbitrage and
equilibrium---are then characterized in terms of standard
differential geometric constructions---such as curvature---associated
with a natural connection in this fibre bundle.
Principal fibre bundle theory has been heavily exploited in
theoretical physics as the language in which laws of nature can be
best formulated by providing an invariant framework to describe
physical systems and their dynamics. These ideas can be carried
over to mathematical finance and economics. A~market is a
financial-economic system that can be described by an appropriate
principal fibre bundle. A~principle as the invariance of market
laws under a change of num\'{e}raire can be seen then as gauge
invariance. Concepts such as No-Free-Lunch-with-Vanishing-Risk (NFLVR) and No-Unbounded-Profit-with-Bounded-Risk (NUPBR) have a geometric characterization, which have
the Capital Asset Pricing Model (CAPM) as a~consequence.\par
The idea that gauge theories are the natural language
for describing economics was first proposed by Malaney and Weinstein
in the context of the economic index problem 
 \cite{Ma96,We06}. Ilinski \cite{Il00,Il01} and Young
\cite{Yo99} proposed to view arbitrage as the curvature of a
gauge connection, in~analogy to some physical theories.
Independently, Cliff and Speed \cite{SmSp98} further developed
Flesaker and Hughston's seminal work \cite{FlHu96} and utilized
techniques from differential geometry to reduce the complexity of asset models before
stochastic~modeling.\par
In recent years, the research on market models admitting arbitrage has started to receive some well-deserved attention, see~\cite{Ru13,HuPr15}. With~this paper, we aim to show that by means of the stochastic differential geometric approach, we can obtain true new results that were not accessible with classical~methods.

This paper is structured as follows. Section~\ref{section2} reviews classical stochastic finance and the Geometric Arbitrage Theory. Arbitrage is seen as the curvature of a principal fibre bundle representing the market which defines the
quantity of arbitrage associated with it. Proofs are omitted and can be found in~\cite{Fa15} and in~\cite{FaTa20}, where the Geometric Arbitrage Theory has been given a rigorous mathematical foundation utilizing the formal background of stochastic differential geometry as in
Elworthy \cite{El82}, Em\'{e}ry\cite{Em89}, Hackenbroch and Thalmaier \cite{HaTh94}, Hsu \cite{Hs02}, Schwartz \cite{Schw80} and
Stroock \cite{St00}.

In Section~\ref{section3} the relationship between arbitrage and spectral theory  is investigated.
The vector bundle associated with the principal fibre bundle represents the cash flow streams corresponding to the assets and carries a
covariant differentiation induced by the connection. The~connection Laplacian
under the Neumann boundary condition is a self-adjoint operator whose spectrum
contains $0$ if and only if the market model satisfies the NFLVR
condition. If~$0$ has simple multiplicity, then the market is complete, and~vice versa. The~eigenspace of the eigenvalue $0$
contains all candidates for the Radon--Nikodym derivative of a possible risk-neutral measure with respect to the statistical~measure.

The Jarrow--Protter--Shimbo theory of asset bubbles for complete no-arbitrage markets is extended to markets allowing for arbitrage opportunities. The classification of bubbles and their decomposition results 
 are proved for models not necessarily satisfying the NFLVR condition. The~connections with Platen--Heath real-world pricing are~highlighted.

In Section~\ref{section4} the cash flow bundle twisted with the exterior algebra of the asset nominal space is a Dirac bundle on which we can apply the Atiyah--Singer index theorem,
which takes the form of the Gauss--Bonnet--Chern theorem for manifolds with boundary, from~which we can infer a topological obstruction to the NFLVR condition, the~non-vanishing of the Euler characteristic of the asset nominal space.
By means of the Bochner-Weitzenb\"ock formula, we obtain another characterization of NFLVR, namely the presence of $0$ in the discrete spectrum of the Dirac Laplacian for the twisted cash flow bundle.
Moreover, we obtain another topological obstruction to NFLVR, namely the non-vanishing of the homology group of the cash flow~bundle.
\par Section~\ref{section5} concludes and
Appendix \ref{Derivatives} reviews Nelson's stochastic derivatives.  

\section{Geometric Arbitrage Theory~Background}\label{section2}
In this section, we explain the main concepts of the Geometric Arbitrage Theory introduced
in~\cite{Fa15}, to~which we refer for proofs and~examples.
\subsection{The Classical Market~Model}\label{StochasticPrelude}
In this subsection, we summarize the classical set up, which
will be rephrased in Section~\ref{foundations}  in differential
geometric terms. We basically follow~\cite{HuKe04} and the ultimate
reference~\cite{DeSc08}.\par We assume continuous-time trading and
that the set of trading dates is $[0,+\infty[$
. This assumption is
general enough to embed the cases of finite and infinite discrete
times as well as one with a finite horizon in continuous time.
Note that while it is true that in the real world trading occurs only at
discrete times, these are not known a priori and can
virtually be any point in the time continuum. This motivates the
technical effort of continuous-time stochastic finance.\par The
uncertainty is modelled by a filtered probability space
$(\Omega,\mathcal{A}, \mathbb{P})$, where $\mathbb{P}$ is the
statistical (physical) probability measure,
$\mathcal{A}=\{\mathcal{A}_t\}_{t\in[0,+\infty[}$ an increasing
family of sub-$\sigma$-algebras of $\mathcal{A}_{\infty}$, and
$(\Omega,\mathcal{A}_{\infty}, \mathbb{P})$ is a probability space.
The filtration $\mathcal{A}$ is assumed to satisfy the usual
conditions, that~is:
\begin{itemize}
\item right continuity: $\mathcal{A}_t=\bigcap_{s>t}\mathcal{A}_s$ for all $t\in[0,+\infty[$;
\item $\mathcal{A}_0$ contains all null sets of $\mathcal{A}_{\infty}$.
\end{itemize}

The market consists of finitely many \textbf{assets} 
 indexed by
$j=1,\dots,N$, whose \textbf{nominal prices} are given by the
vector valued semimartingale $S:[0,+\infty[\times\Omega\rightarrow\mathbb{R}^N$,
denoted by $(S_t)_{t\in[0,+\infty[}$ and adapted to the filtration $\mathcal{A}$.
The stochastic process $(S^ j_t)_{t\in[0,+\infty[}$ describes the
price at time $t$ of the $j$th asset in terms of  unit of cash
\textit{at time $t=0$}. More precisely, we assume the existence of a
$0$th asset---the~\textbf{cash}, a~strictly positive
semimartingale that evolves according to
$S_t^0=\exp(\int_0^tdu\,r^0_u)$, where the integrable
semimartingale $(r^0_t)_{t\in[0,+\infty[}$ represents the
continuous interest rate provided by the cash account. One always
knows in advance what the interest rate on one's own bank account
is, but~this can change from time to time. The~cash account is
therefore considered the locally riskless asset in contrast to
the other assets, the~risky ones. Subsequently, we will mainly
utilize \textbf{discounted prices}, defined as
$\hat{S}_t^j:=S_t^j/S^{0}_t$, representing the asset prices in
terms of a \textit{current} unit of~cash.\par
 We remark that there is no need to
assume that asset prices are positive. However, there must be at least
one strictly positive asset, which in~our case is the cash. If~we want to
renormalize the prices by choosing another asset instead of the
cash as reference, i.e.,~by making it our
\textbf{num\'{e}raire}, then this asset must have a strictly
positive price process. More precisely, a~generic num\'{e}raire
is a portfolio of the original assets $j=0,1,2,\dots,N$, whose nominal price is represented by a strictly
positive stochastic process $(B_t)_{t\in[0,+\infty[}$. The~discounted prices of the original
assets are  then represented in terms of the num\'{e}raire by the
semimartingales $\hat{S}_t^j:=S_t^j/B_t$.\par We assume that there
are no transaction costs and that short sales are allowed. Remark
that the absence of transaction costs can be a serious limitation
for a realistic model. The~filtration $\mathcal{A}$ is not
necessarily generated by the price process
$(S_t)_{t\in[0,+\infty[}$; sources of information other than
prices are allowed. All agents have access to the same information
structure, that is, to the filtration $\mathcal{A}$.\par
Let $v$ be a positive real number.
A $v$-admissible \textbf{strategy} $x=(x_t)_{t\in[0,+\infty[}$ is a predictable $S$-integrable process for which the It\^{o} integral $\int_0^tx\cdot dS\ge-v$ a.s. for all $t\ge0$ with $x_0=0$. A~strategy is admissible if it is $v$-admissible for some $v\ge0$.

\begin{defi}[\textbf{Arbitrage}] Let the process $(S_t)_{[0,+\infty[}$ be a semimartingale and $(x_t)_{t\in[0,+\infty[}$ be an admissible self-financing strategy. Let us consider trading up to time $T\le\infty$. The~portfolio wealth at time $t$ is given by $V_{t}(x):=V_0+\int_0^tx_u\cdot dS_u$, and~we denote as $K_0$ the subset of $L^0(\Omega, \mathcal{A}_{T},P)$ containing all such $V_T(x)$, where $x$ is any admissible self-financing strategy.
We~define the following:
\begin{itemize}
\item $C_0:=K_0-L_+^0(\Omega, \mathcal{A}_{T},P)$;
\item $C:=C_0\cap L^{\infty}(\Omega, \mathcal{A}_{T},P)$;
\item $\bar{C}$: the closure of $C$ in $L^{\infty}$ with respect to the norm topology;
\item $\mathcal{V}^{V_0}:=\left\{(V_{t})_{t\in[0,+\infty[}\,\big{|}\, V_t=V_t(x), \,\text{where } x \text{ is } V_0\text{-admissible} \right\}$;
\item $\mathcal{V}_T^{V_0}:=\left\{V_T\,\big{|}\,(V_{t})_{t\in[0,+\infty[}\in\mathcal{V}^{V_0}\right\}$; \\ terminal wealth for $V_0$-admissible self-financing strategies.
\end{itemize}
And let $L_+^{\infty}(\Omega, \mathcal{A}_{T},P)$ be the set of positive random variables in $L^{\infty}(\Omega, \mathcal{A}_{T},P)$. We say that $S$ satisfies:
\begin{itemize}
\item \textbf{(NA) no-arbitrage} if~and only if $C \cap L_+^{\infty}(\Omega, \mathcal{A}_{T},P)=\{0\}$;
\item \textbf{(NFLVR) no-free-lunch-with-vanishing-risk}  if~and only if $\bar{C} \cap L_+^{\infty}(\Omega, \mathcal{A}_{T},P)=\{0\}$;
\item \textbf{(NUPBR) no-unbounded-profit-with-bounded-risk} if~and only if $\mathcal{V}_T^{V_0}$ is \\ bounded in $L^0$ for some $V_0>0$.
\end{itemize}
\end{defi}

 The relationship between these three different types of arbitrage has been elucidated in~\cite{DeSc94} and in~\cite{Ka97}, with the proof of the following result.
\begin{theorem}
\begin{equation}
\text{(NFLVR)}\Leftrightarrow \text{(NA)}+\text{(NUPBR)}.
\end{equation}
\end{theorem}
\begin{theorem}[\textbf{First fundamental theorem of asset pricing}]
The market $(S,\mathcal{A})$ satisfies the NFLVR condition if and only if there exists an equivalent local martingale measure $P^*$.
\end{theorem}
\begin{remark}
In the first fundamental theorem of asset pricing, we just assumed that the price process $S$ is locally bounded. If~$S$ is bounded, then NFLVR is equivalent to the existence of a martingale measure.
But without this additional assumption, NFLVR only implies the existence of a \textit{local} martingale measure, i.e.,~a local martingale which is \textit{not} a martingale. This distinction is important because~the difference between a security price process being a strict local martingale versus a martingale under a probability $P^*$ is related to the existence of asset price bubbles.
\end{remark}
\begin{defi}[\textbf{Complete market}]
The market $(S,\mathcal{A})$ is \textbf{complete} on $[0,T]$ if for all contingent claims $C\in L_{+}(P^*,\mathcal{A}_T):=\{C:\Omega\rightarrow[0,+\infty[|\,C\text{ is }\mathcal{A}_T-\text{measurable, and }\mathbb{E}_0^{P^*}[|C|]<+\infty\}$ there exists an admissible self-financing strategy $x$ such that $C=V_T(x)$.
\end{defi}
\begin{theorem}[\textbf{Second fundamental theorem of asset pricing}]
Given that $(S,\mathcal{A})$ satisfies the NFLVR condition, the~market is complete on $[0,T]$ if and only if the equivalent local martingale $P^*$ is unique.
\end{theorem}
\begin{defi}[\textbf{Dominance}]
The j-th security $S^j=(S^j_t)_{t\in[0.T]}$ is \textbf{undominated} on $[0,T]$ if there is no admissible strategy $(x_t)_{t\in[0,T]}$ such that:
\begin{equation}
S^j_0+ (x\cdot S)_T\ge S^j_T\quad\text{ a.s. and }\quad P[S^j_0+ (x\cdot S)_T > S^j_T]>0.
\end{equation}
\noindent We say that $S$ satisfies the \textbf{no dominance} condition (ND) on $[0,T]$ if and only if each $S^j$, $j=0,1,\dots,N$ is undominated on $[0,T]$.
\end{defi}

\begin{defi}[\textbf{Economy}]
An \textbf{economy} consists of a market given by $(S,\mathcal{A})$ and a finite number of investors $k=1,\dots,K$ characterized by their beliefs, information, preferences, and endowment. Moreover, there is a single consumption good that is perishable. The~price of the consumption good in units of the cash account is denoted as $\Psi=(\Psi_t)_{t\in[0,T]}$. We assume that $\Psi$ is strictly positive. The~$k$-th investor is characterized by the following~quantities:
\begin{itemize}
 \item\textbf{Beliefs and information: } $(P_k,\mathcal{A})$. We assume that the investor's beliefs $P_k$ are equivalent to $P$. All investors have the same information filtration $\mathcal{A}$.
 \item\textbf{Utility function: } $U_k:[0,T]\times[0,+\infty[\rightarrow\mathbb{R}$ and $\mu$ is a probability measure on $[0,T]$ with $\mu(\{T\})>0$ such that for every $t$ in the support of $\mu$, the~function $U_k(t,\cdot)$ is strictly increasing. We also assume $\lim_{v\rightarrow +\infty}U_k(T,v)=+\infty$. The utility that agent $k$ derives from consuming $c_t\mu(dt)$ at each time $t\le T$ is as follows:
\begin{equation}
       \mathcal{U}_k(c)=\mathbb{E}^k_0\left[\int_0^TU_k(t,c_t)\mu(dt)\right],
     \end{equation}
     \noindent where $\mathbb{E}^k$ is the expectation with respect to $P_k$. Since $\mu(\{T\})>0$, the utility is strictly increasing in the final consumption $c_T$.
 \item\textbf{Initial wealth: } $v_k$. Given a trading strategy $x=(x^1,\cdots,x^N)$, the~investor will be required to choose his initial holding $x^0_0$ in the cash account such that:
\begin{equation}
     v_k=x^0_0+\sum_{j=1}^Nx^j_0S^i_0.
   \end{equation}
\item\textbf{Stochastic endowment stream: } $\epsilon^k_t$, $t<T$ of the commodity. This means that the investors receive  $\epsilon^k_t\mu(dt)$ units of the commodity at time $t\le T$. The~cumulative endowment of the k-th investor in~units of the cash account is given by the following equation:
\begin{equation}
      \mathcal{E}^k_t:=\int_0^t\Psi_s\epsilon^k_s\mu(ds).
    \end{equation}
\end{itemize}
\end{defi}

\begin{defi}[\textbf{Consumption plan and strategy}] A pair $(c^k_t,x^k_t)_{t\in[0,T]}$ is called \textbf{admissible} if $(c^k_t)_{t\in[0,T]}$ is progressively measurable with respect to the filtration $\mathcal{A}$, $(x^k_t)_{t\in[0,T]}$ is admissible in the usual sense, and~it generates a wealth process $V^k=(V^K_t)_{t\in[0,T]}$ with non-negative terminal wealth $V^k_T\ge0$.
\end{defi}

\begin{defi}[\textbf{Equilibrium}]\label{equilibrium}
Given an economy $(\{P_k\}_{k=1,\dots,K}, (\mathcal{A}_t)_{t\in[0,T]},\\ \{\epsilon_k\}_{k=1,\dots,K},$ $\{U_k\}_{k=1,\dots,K})$, a~consumption good price index $\Psi$, financial assets $S=[S^0,S^1,\dots,S^N]^{\dagger}$, and~investor consumption-investment plans $(\hat{c}^k,\hat{x}^k)$ for $k=1,\dots K$, the~pair $(\Psi,S)$ is an \textbf{equilibrium price process} if for all $t\le T$ $P-$a.s. the following conditions are~satisfied:
\begin{itemize}
\item\textbf{Securities markets clear:}
\begin{equation}
   \sum_{k=1}^K\hat{x}^{k,j}_t=\alpha^j\qquad(j=0,1,\dots,N),
  \end{equation}
\noindent where $\alpha^j$ is the aggregate net supply of the j-th security. It is assumed that each $\alpha^j$ is non-random and constant over time, with~$\alpha^0=0$ and $\alpha^j>0$ for $j=1,\dots,N$.
\item\textbf{Commodity markets clear: }
\begin{equation}
   \sum_{k=1}^K\hat{c}^k_t= \sum_{k=1}^K\epsilon^k_t.
  \end{equation}
\item\textbf{Investors' choices are optimal:} $(\hat{c}^k, \hat{x}^k)$ solves the k-th investor's utility maximization problem
\begin{equation}
    u_k(x):=\sup\{U_k(c)|\, c \text{ admissible consumption plan}, x^k =x\},
    \end{equation}
    \noindent and the optimal value is finite.
\end{itemize}
\end{defi}
\begin{defi}[\textbf{Efficiency}]
A market model given by $S$ is called \textbf{efficient} on $[0,T]$ with respect to $(\mathcal{A}_t)_{t\in[0,T]}$, i.e.,~(E), if~there exists a consumption good price index $\Psi$ and an economy $(\{P_k\}_{k=1,\dots,K}, (\mathcal{A}_t)_{t\in[0,T]}, \{\epsilon_k\}_{k=1,\dots,K}, \{U_k\}_{k=1,\dots,K})$, for~which $(\Psi,S)$ is an equilibrium price process on $[0,T]$.
\end{defi}

 In~\cite{JaLa12} and ~\cite{Ja12} we find the proof of the following result. 

\begin{theorem}[\textbf{Third fundamental theorem of asset pricing, characterization of efficiency}]
Let $(S, \mathcal{A})$ be a market. The~following statements are~equivalent:
\begin{itemize}
\item[(i)] (E): $(S, \mathcal{A})$ is efficient in $[0,T]$;
\item[(ii)] $(S, \mathcal{A})$ satisfies both (NFLVR) and (ND) on $[0,T]$;
\item[(iii)] (EMM): There exists a probability $P^*$ equivalent to $P$ such that $S$ is a $(P^*,\mathcal{A})$ martingale on $[0,T]$.
\end{itemize}
\end{theorem}

\subsection{Geometric Reformulation of the Market Model: Primitives}
We are going to introduce a more general representation of the
market model introduced in Section \ref{StochasticPrelude}, which
better suits the arbitrage modeling task.
\begin{defi}\label{defi1}
A \textbf{gauge} is an ordered pair of two $\mathcal{A}$-adapted
real valued semimartingales $(D, P)$, where
$D=(D_t)_{t\ge0}:[0,+\infty[\times\Omega\rightarrow\mathbb{R}$ is
called a \textbf{deflator} and
$P=(P_{t,s})_{t,s}:\mathcal{T}\times\Omega\rightarrow\mathbb{R}$,
which is called a \textbf{term structure}, is considered a
stochastic process with respect to the time $t$, termed
\textbf{valuation date}, and
$\mathcal{T}:=\{(t,s)\in[0,+\infty[^2\,|\,s\ge t\}$. The~parameter
$s\ge t$ is referred to as \textbf{maturity date}. The~following
properties must be satisfied a.s. for all $t, s$ such that $s\ge
t\ge 0$:
 \begin{itemize}
  \item [(i)] $P_{t,s}>0$;
  \item [(ii)] $P_{t,t}=1$.
 \end{itemize}
\end{defi}

\begin{remark}
Deflators and term structures can be considered
\textit{outside the context of fixed income.} An arbitrary
financial instrument is mapped to a gauge $(D, P)$ with the
following economic~interpretation:
\begin{itemize}
\item Deflator: $D_t$ is the value of the financial instrument at time $t$ expressed in terms of some num\'{e}raire. If~we
choose the cash account, the~$0$-th asset, as num\'{e}raire, then
we can set $D_t^j:=\hat{S}_t^j=\frac{S_t^j}{S_t^0}\quad(j=1,\dots
N)$.
\item Term structure: $P_{t,s}$ is the value at time $t$ (expressed in units of
deflator at time $t$) of a synthetic zero coupon bond with
maturity $s$ delivering one unit of financial instrument at time
$s$. It represents a term structure of forward prices with respect
to the chosen num\'{e}raire.
\end{itemize}

 We point out that there is no unique choice for
deflators and term structures describing an asset model. For~example, if~a set of deflators qualifies, then we can multiply
every deflator  by the same positive semimartingale to obtain
another suitable set of deflators. Of~course, term structures have
to be modified accordingly. The~term ``deflator'' is clearly
inspired by actuarial mathematics. In~the present context, it
refers to a nominal asset value division by a strictly positive
semimartingale (which can be the state price deflator, if this
exists, and it is made the num\'{e}raire). There is no need to
assume that a deflator is a positive process. However, if~we want
to make an asset our num\'{e}raire, then we have to make sure
that the corresponding deflator is a strictly positive stochastic
process.
\end{remark}

\subsection{Geometric Reformulation of the Market Model: Portfolios}\label{trans}
We would now like to introduce transforms of deflators and term structures
in order to group gauges containing the same (or less) stochastic
information. In this regard, we will consider \textit{deterministic}
linear combinations of assets modelled by the same gauge (e.g., zero
bonds of the same credit quality with different maturities).

\begin{defi}\label{gaugeTransforms2}
Let $\pi:[0, +\infty[\longrightarrow \mathbb{R}$ be a deterministic
cash flow intensity (possibly generalized) function. It induces a
\textbf{gauge transform} $(D,P)\mapsto
\pi(D,P):=(D,P)^{\pi}:=(D^{\pi}, P^{\pi})$ by the following formula:
\begin{equation}
 D_t^{\pi}:=D_t\int_0^{+\infty}dh\,\pi_h P_{t, t+h}\qquad
P_{t,s}^{\pi}:=\frac{\int_0^{+\infty}dh\,\pi_h P_{t,
s+h}}{\int_0^{+\infty}dh\,\pi_h P_{t, t+h}}.
\end{equation}
\end{defi}

\begin{proposition}\label{conv}
Gauge transforms induced by cash flow vectors have the following
property:
\begin{equation}((D,P)^{\pi})^{\nu}= ((D,P)^{\nu})^{\pi} = (D,P)^{\pi\ast\nu},\end{equation} where
$\ast$ denotes the convolution product of two cash flow vectors or
intensities respectively:
\begin{equation}\label{convdef}
    (\pi\ast\nu)_t:=\int_0^tdh\,\pi_h\nu_{t-h}.
\end{equation}
\end{proposition}

 The convolution of two non-invertible gauge transforms is
non-invertible. The~convolution of a non-invertible with an
invertible gauge transform is~non-invertible.

\begin{defi}\label{int}  If the term structure is differentiable with respect to the maturity date, it can be written
 as a functional of the
\textbf{instantaneous forward rate } f defined as follows:
\begin{equation}
  f_{t,s}:=-\frac{\partial}{\partial s}\log P_{t,s},\quad
  P_{t,s}=\exp\left(-\int_t^sdhf_{t,h}\right),
\end{equation}
\noindent and
\begin{equation}
 r_t:=\lim_{s\rightarrow t^+}f_{t,s}
\end{equation}
\noindent is termed \textbf{short rate}.
\end{defi}

\begin{remark} The special choice of vanishing interest rate $r\equiv0$ or flat term structure
$P\equiv1$ for all assets corresponds to the classical model,
where only asset prices and their dynamics are relevant.
\end{remark}

\subsection{Arbitrage Theory in a Differential Geometric~Framework}\label{foundations} Now we are in the position to
rephrase the asset model presented in Section
\ref{StochasticPrelude} in terms of a natural geometric language.
Given $N$ base assets, we want to construct a
portfolio theory and study arbitrage, and thus we cannot a priori assume the existence of a
risk-neutral measure or a state price deflator. In~terms of
differential geometry, we will adopt the mathematician's and not
the physicist's approach. The~market model is seen as a principal
fibre bundle of the (deflator, term structure) pairs, discounting
and foreign exchange as a parallel transport, num\'{e}raire as the
global section of the gauge bundle, and arbitrage as curvature.  The~no-free-lunch-with-vanishing-risk condition is proved to be
equivalent to a zero-curvature~condition.

\subsubsection{Market Model as Principal Fibre~Bundle}
Let us consider, in continuous time, a market with $N$ assets and a num\'{e}raire. A~general portfolio at time $t$ is described by the vector of nominals $x\in \mathfrak{X}$, for~an open set $\mathfrak{X}\subset\mathbb{R}^N$. By~nominals $x^1,\dots,x^N$ we mean the number of assets that we hold in our portfolio. Following Definition \ref{defi1}, the~asset model consisting in $N$ synthetic zero bonds is described by means of the following gauges:
\begin{equation}(D^j,P^j)=((D_t^j)_{t\in[0, +\infty[},(P_{t,s}^j)_{s\ge t}),\end{equation}
\noindent where $D^j$ denotes the deflator and $P^j$ the term
structure. This can be written as follows:
\begin{equation}P_{t,s}^j=\exp\left(-\int_t^sf^j_{t,u}du\right),\end{equation}
where $f^j$ is the instantaneous forward rate process for the $j$-th asset and the corresponding short rate is given by $r_t^j:=\lim_{u\rightarrow 0^+}f^j_{t,u}$. For~a
portfolio with nominals $x\in \mathfrak{X}\subset\mathbb{R}^N$, we define:
\begin{equation}
D_t^x:=\sum_{j=1}^Nx_jD_t^j\qquad
f_{t,u}^x:=\sum_{j=1}^N\frac{x_jD_t^j}{\sum_{j=1}^Nx_jD_t^j}f_{t,u}^j\qquad
P_{t,s}^x:=\exp\left(-\int_t^sf^x_{t,u}du\right).
\end{equation}

The short rate writes as follows:
\begin{equation}
r_t^x:=\lim_{u\rightarrow 0^+}f^x_{t,u}=\sum_{j=1}^N\frac{x_jD_t^j}{\sum_{j=1}^Nx_jD_t^j}r_t^j.
\end{equation}

The image space of all possible strategies reads as follows:
\begin{equation}M:=\{(t,x)\in [0,+\infty[\times\mathfrak{X}\}.\end{equation}

In Section \ref{trans}, cash flow intensities and the corresponding
gauge transforms were introduced. They have the structure of an
Abelian semigroup:
\begin{equation}
 H:=\mathcal{E}^{\prime}([0,
+\infty[,\mathbb{R})=\{F\in\mathcal{D}^{\prime}([0,+\infty[)\mid
\text{supp}(F)\subset[0, +\infty[\text{ is compact}\},
\end{equation}
where the semigroup operation on distributions with compact support
is the convolution (see~\cite{Ho03}, Chapter IV), which extends the
convolution of regular functions as defined by Formula
(\ref{convdef}).
\begin{defi}\label{MFB}
The \textbf{Market Fibre Bundle} is defined as the fibre bundle of the following
gauges:
\begin{equation}
\mathcal{B}:=\{ ({D^x_t},{P^x_{t,\,\cdot}})^{\pi }|\,(t,x)\in
M, \pi\in G\}.
\end{equation}
\end{defi}
The cash flow intensities defining invertible transforms constitute
an Abelian group:
\begin{equation}
G:=\{\pi\in H|\text{ it exists } \nu\in H\text{ such that
}\pi\ast\nu=\delta\}\subset \mathcal{E}^{\prime}([0,
+\infty[,\mathbb{R}).
\end{equation}
From Proposition \ref{conv}, we obtain the following theorem. 

\begin{theorem} The market fibre bundle $\mathcal{B}$ has the
structure of a $G$-principal fibre bundle  given by the action:
\begin{equation}
\begin{split}
\mathcal{B}\times G &\longrightarrow\mathcal{B}\\
 ((D,P), \pi)&\mapsto (D,P)^{\pi}=(D^{\pi},P^{\pi})
\end{split}
\end{equation}
The group $G$ acts freely and differentiably on
$\mathcal{B}$ to the right.
\end{theorem}

\subsubsection{Nelson $\mathcal{D}$ Weak Differentiable Market~Model} We continue to reformulate the classic asset model introduced in Section \ref{StochasticPrelude} in terms of stochastic differential geometry.
\begin{defi}\label{weakMM}
 A \textbf{Nelson $\mathcal{D}$ weak differentiable market model} for $N$ assets is described by $N$ gauges, which are Nelson $\mathcal{D}$ weak differentiable with respect to the time variable. More exactly, for~all $t\in[0,+\infty[$ and $s\ge t$, there is an open time interval $I\ni t$ such that for the deflators $D_t:=[D_t^1,\dots,D_t^N]^{\dagger}$ and the term structures $P_{t,s}:=[P_{t,s}^1,\dots,P_{t,s}^N]^{\dagger}$, the~latter seen as processes in $t$ and parameter $s$, there exists a $\mathcal{D}$ weak $t$-derivative (see Appendix \ref{Derivatives}). The~short rates are defined by $r_t:=\lim_{s\rightarrow t^{-}}\frac{\partial}{\partial s}\log P_{ts}$.\par
 A strategy is a curve $\gamma:I\rightarrow X$ in the portfolio space parameterized by the time. This means that the allocation at time $t$ is given by the vector of nominals $x_t:=\gamma(t)$. We denote as $\bar{\gamma}$ the lift of $\gamma$ to $M$, that is $\bar{\gamma}(t):=(\gamma(t),t)$. A~strategy is said to be \textbf{closed} if it represented by a closed curve.  A~\textbf{weak $\mathcal{D}$-admissible strategy} is predictable and $\mathcal{D}$- weak differentiable.
\end{defi}
\begin{remark}
We require weak $\mathcal{D}$-differentiability and not strong $\mathcal{D}$-differentiability because imposing a priori regularity properties on the trading strategies corresponds to restricting the class of admissible strategies with respect to the classical notion of Delbaen
and Schachermayer. Every (no-)arbitrage consideration depends crucially on the chosen definition
of admissibility. Therefore, restricting the class of admissible strategies may lead to
the automatic exclusion of potential arbitrage opportunities, leading to vacuous statements of FTAP-like results. An~admissibile strategy in the classic sense (see Section~\ref{section2}) is weak $\mathcal{D}$-differentiable.
\end{remark}

 In general the allocation can depend on the state of the nature, i.e.,~$x_t=x_t(\omega)$ for $\omega\in\Omega$.
\begin{proposition}
A weak $\mathcal{D}$-admissible strategy is self-financing if and only if:
\begin{equation}\label{sf}
\mathcal{D}(x_t\cdot D_t)=x_t\cdot \mathcal{D}D_t-\frac{1}{2}\mathfrak{D}_*\left<x,D\right>_t\text{ or }
\mathcal{D}x_t\cdot D_t=-\frac{1}{2}\mathfrak{D}_*\left<x,D\right>_t\text{ or }
\mathfrak{D}x_t\cdot D_t=0,
\end{equation}
almost surely. The~bracket $\left<\cdot,\cdot\right>$ denotes the continuous part of the quadratic covariation.
\end{proposition}

For the remainder of this paper, unless otherwise stated, we will only deal with weak $\mathcal{D}$ differentiable market models, weak $\mathcal{D}$ differentiable strategies, and, when necessary, with~weak $\mathcal{D}$ differentiable state price deflators. All It\^{o} processes are weak $\mathcal{D}$ differentiable, so that the class of considered admissible strategies is very~large.

\subsubsection{Arbitrage as~Curvature}
The Lie algebra of $G$ is the function space of all real valued functions on $[0, +\infty [$ denoted as:
\begin{equation}\mathfrak{g}=\mathbb{R}^{[0, +\infty[}\end{equation}
and therefore commutative.
Following Ilinski's idea~\cite{Il01}, we motivate the choice of a particular $\mathfrak{g}$-valued connection $1$-form by the fact that it allows us to encode portfolio rebalancing (or foreign exchange) and discounting as parallel transport.
\begin{theorem}\label{Ilinski}With the choice of connection
\begin{equation}\label{connection}\chi(x,t,g).(\delta x, \delta t):= \left(\frac{D_t^{\delta x}}{D_t^x}-r_t^x\delta t\right) g,\end{equation}
the stochastic parallel transport in $\mathcal{B}$ has the following financial~interpretations:
\begin{itemize}
\item Parallel transport along the nominal directions ($x$-lines) corresponds to a multiplication by an exchange rate;
\item Parallel transport along the time direction ($t$-line) corresponds to a division by a stochastic discount factor.
\end{itemize}
\end{theorem}
\begin{proof}
We refer to Theorem 28 in~\cite{Fa15}.
\end{proof}

Recall that the time derivatives needed to define the parallel transport along the time lines have to be understood in Stratonovich's sense. We see that the bundle is trivial, because~it has a global trivialization, but~the connection is not trivial.
 The connection $\chi$ writes
 as a linear combination of basis differential forms:
\begin{equation}
\chi(x,t,g)=\left(\frac{1}{D_t^x}\sum_{j=1}^ND_t^jdx_j-r_t^xdt\right)g.
\label{1-form}
\end{equation}

The ${g}$-valued curvature $2$-form is defined as follows:
\begin{equation}R:=d\chi+[\chi,\chi],\end{equation} which means that by this, for all $(x,t,g)\in \mathcal{B}$ and for all $\xi,\eta\in T_{(x,t)}M$,
\begin{equation}R(x,t,g)(\xi,\eta):=d\chi(x,t,g)(\xi,\eta)+[\chi(x,t,g)(\xi),\chi(x,t,g)(\eta)]. \end{equation}

Note that the Lie algebra being commutative, the~Lie bracket $[\cdot,\cdot]$ vanishes. After~some calculations, we obtain the following:
\begin{equation}\label{defCurv}R(x,t,g)=\frac{g}{D_t^x}\sum_{j=1}^ND_t^j\left(r_t^x+\mathcal{D}\log(D_t^x)-r_t^j-\mathcal{D}\log(D_t^j)\right)dx_j\wedge dt,\end{equation}
summarized as the following proposition. 
\begin{proposition}[\textbf{Curvature Formula}]\label{curvature}
Let $R$ be the curvature. Then, the~following quality holds:
\begin{equation}R(t,x,g)=g dt\wedge d_x\left[\mathcal{D} \log (D_t^x)+r_t^x\right].\end{equation}
\end{proposition}
The following result  characterizes arbitrage as curvature.
\begin{theorem}[\textbf{No-Arbitrage}]\label{holonomy}
The following assertions are~equivalent:
\begin{itemize}
\item [(i)] The market model (with base assets and futures with discounted prices $D$  and $P$) satisfies the no-free-lunch-with-vanishing-risk condition;
\item[(ii)] There exists a positive martingale $\beta=(\beta_t)_{t\ge0}$ such that deflators and short rates satisfy, for all portfolio nominals and all times, the condition
\begin{equation}r_t^x=-\mathcal{D}\log(\beta_tD_t^x);\end{equation}
\item[(iii)] There exists a positive martingale $\beta=(\beta_t)_{t\ge0}$ such that deflators and term structures satisfy, for all portfolio nominals and all times, the condition
\begin{equation}P^x_{t,s}=\frac{\mathbb{E}_t[\beta_sD^x_s]}{\beta_tD^x_t}.\end{equation}
\end{itemize}
\end{theorem}

 This motivates the following definition.
\begin{defi}
The market model satisfies the \textbf{zero curvature (ZC)} if and only if
the curvature vanishes a.s.
\end{defi}

\noindent Therefore, we have following implications relying on two different definitions of no-arbitrage:
\begin{corollary}
\begin{equation}
{(NFLVR)}\Rightarrow {(ZC)}.
\end{equation}
\end{corollary}

 As proved in~\cite{FaTa20}, the~two weaker notions of arbitrage---the~zero curvature and the no-unbounded-profit-with-bounded-risk---are satisfied.
\begin{theorem}\label{thm_ZC_equiv}
\begin{equation}
{(NUPBR)}\Rightarrow{(ZC)}.
\end{equation}
\end{theorem}

 The converse is true for special cases of It\^{o}'s dynamics for asset values and term structures (see~\cite{FaTa20}).
\begin{remark}\label{rem_holonomy}
Let us consider some special cases of Theorem \ref{holonomy}.
\begin{enumerate}
\item \textbf{The components of $r$ are equal: } For example, in the classical model, where there are no term structures (i.e., $r\equiv0$),
\begin{enumerate}
\item \textbf{$D$ and $r$ are constant over time: } NFLVR is satisfied;
\item \textbf{$D$ and $r$ are deterministic and not constant over time:} NFLVR is never satisfied.
\end{enumerate}
\item \textbf{The components of $r$ are not equal: }
\begin{enumerate}
\item \textbf{$D$ and $r$ are constant over time: }NFLVR is never satisfied;
\item \textbf{$D$ and $r$ are deterministic and not constant over time:} NFLVR can be satisfied if (ii) or (iii) hold true.
\end{enumerate}
\end{enumerate}
\end{remark}
\subsubsection{Expected Utility Maximization and the CAPM~Formula}
\begin{defi}[\textbf{EUM}]
The \textbf{expected utility maximization}  problem for the final wealth over the period $[0,t]$ for a given utility function $u$ reads as follows:
\begin{equation}\label{eq}
\boxed{
\max_{\substack{(x_u)_{u\in[0,t]}\\(x_u)_u\text{ is self-financing and admissible}\\x_0\cdot D_0=1}}\mathbb{E}_0\left[u(x_t\cdot D_t)\right].
}
\end{equation}
We denote as (EUM) the existence of a solution and its uniqueness for (\ref{eq}).
\end{defi}

 As proved in~\cite{FaTa20}, e.g., the~existence of a solution for an expected utility maximization problem and the no-unbounded-profit-with-bounded-risk is satisfied.
\begin{theorem}
\begin{equation}
{(NUPBR)}\Rightarrow{(EUM)}
\end{equation}
\end{theorem}

 Asset returns and the market portfolio return are related at the expected value level by means of sensitivities. The~relationship is the following equilibrium result.
\begin{theorem}[\textbf{CAPM}]
Let
\begin{equation}
R_{[0,t]}:=\frac{D_t}{D_0}-1\qquad\qquad R_{[0,t]}^M:=\frac{D_t^{x_t^M}}{D_0^{x_0^M}}-1
\end{equation}
be the discounted assets total returns and the discounted market portfolio total return. If~we assume that the expected utility of the final wealth of every portfolio is maximized, then
\begin{equation}
\boxed{
\mathbb{E}_0[R_{[0,t]}]=\frac{ \Cov_0\left(R_{[0,t]},R_{[0,t]}^M\right)}{ \vari_0\left(R_{[0,t]}^M\right)}\mathbb{E}_0[R_{[0,t]}^M].
}
\end{equation}
\end{theorem}

\begin{proof}
\noindent If we choose a quadratic utility function
\begin{equation}
u(v):=v-\frac{\lambda}{2}v^2,
\end{equation}
\noindent where $\lambda$ denotes the risk aversion parameter, and~we assume that only at times $\{0,t\}$ rebalancing during the interval $[0,t]$ is allowed,
then the expected utility maximization problem becomes
\begin{equation}
\max_{\substack{x_0\\x_0\cdot D_0=1}}\mathbb{E}_0\left[(x_t\cdot D_t)-\frac{\lambda}{2}(x_t\cdot D_t)^2\right],
\end{equation}
\noindent which is equivalent to
\begin{equation}
\max_{\substack{x_0\\x_0\cdot D_0=1}}\mathbb{E}_0\left[(x_t\cdot D_t)-\frac{\lambda}{2}\left((x_t\cdot D_t)^2-\mathbb{E}_0[x_t\cdot D_t]\right)\right],
\end{equation}
\noindent and to
\begin{equation}
\max_{\substack{w_0\\w_0\cdot e=1}}w_0^{\dagger}\mathbb{E}_0[R_{[0,t]}]-\frac{\lambda}{2}w_0^{\dagger}\VCM_0(R_{[0,t]})w_0,
\end{equation}
\noindent where $w_0:=\frac{x_0D_0}{x_0\cdot D_0}$ are the portfolio weights. The~solution is the market portfolio
\begin{equation}\label{sol_MV}
w_0^M=\frac{{\hat{w}}_0^M}{\sqrt{{(\hat{w}}^M_0)^{\dagger}{\hat{w}}_0^M}},
\text{ where }
\quad\hat{w}_0^M:=\frac{1}{\lambda}\VCM_0(R_{[0,t]})^{-1}\mathbb{E}_0[R_{[0,t]}].
\end{equation}

 Therefore, we have by (\ref{sol_MV})
\begin{equation}
\begin{split}
\mathbb{E}_0[R_{[0,t]}]& =\lambda \VCM_0(R_{[0,t]})\hat{w}_0^M=\frac{\lambda}{|\hat{w}_0^M|}\VCM_0(R_{[0,t]})w_0^M=\\
&=\frac{\lambda}{|\hat{w}_0^M|}\Cov_0\left(R_{[0,t]}, R_{[0,t]}^M\right),
\end{split}
\end{equation}
\noindent and by multiplication with $w_0^M$
\begin{equation}\label{exp_rets}
\mathbb{E}_0[R_{[0,t]}^M]=\frac{\lambda}{|\hat{w}_0^M|^2}(w_0^M)^{\dagger}\VCM_0(R_{[0,t]})w_0^M=\frac{\lambda}{|\hat{w}_0^M|^2}\vari_0\left(R_{[0,t]}^M\right),
\end{equation}
\noindent from which we infer that
\begin{equation}
\frac{\lambda}{|\hat{w}_0^M|^2}=\frac{\mathbb{E}_0[R_{[0,t]}^M]}{\vari_0\left(R_{[0,t]}^M\right)},
\end{equation}
\noindent which, inserted into (\ref{exp_rets}), leads to
\begin{equation}
\mathbb{E}_0[R_{[0,t]}]=\frac{\Cov_0\left(R_{[0,t]}, R_{[0,t]}^M\right)}{\vari_0\left(R_{[0,t]}^M\right)}\mathbb{E}_0[R_{[0,t]}],
\end{equation}
\noindent and the proof is completed.
\end{proof}

\begin{remark}
The vector
\begin{equation}
\beta_0:=\frac{ \Cov_0\left(R_{[0,t]}, R_{[0,t]}^M\right)}{\vari_0\left(R_{[0,t]}^M\right)}
\end{equation}
contains the sensitivities of the expected asset returns with respect to the expected market portfolio return. We can compute the CAPM in its classical form for the returns of the asset as follows:
\begin{equation}
r_{[0,t]}:=\frac{S_t}{S_0}-1=\left(1+R_{[0,t]}\right)\exp\left(+\int_0^tdu\,r^0_u\right)-1
\end{equation}
obtaining
\begin{equation}
\boxed{
\mathbb{E}_0[r_{[0,t]}]-r_{[0,t]}^f=\frac{\Cov_0\left(r_{[0,t]}, r_{[0,t]}^M\right)}{\vari_0\left(r_{[0,t]}^M\right)}\left(\mathbb{E}_0[r_{[0,t]}]-r_{[0,t]}^f\right),
}
\end{equation}
\noindent where $r_{[0,t]}^f:=\exp\left(+\int_0^tdu\,r^0_u\right)-1$ is the risk-free return.
\end{remark}
\begin{remark}
The different arbitrage concepts and the Capital Asset Pricing Model are related in the following logical representation:
\begin{equation}\small
\boxed{
\text{(EMM)} \Leftrightarrow\text{(E)} \Leftrightarrow \left\{
                                \begin{array}{ll}
\text{(NFLVR)}\Leftrightarrow \left\{
                                \begin{array}{ll}
                                  \text{(NUPBR)}\Rightarrow \left\{\begin{array}{l}\text{(EUM)}\\\text{(ZC)}\end{array}\right\}\Rightarrow\text{(CAPM)}\\
                                  \text{(NA)}
                                \end{array}
                              \right.\\
\text{(ND)}
   \end{array}
                              \right.
}
\end{equation}
\end{remark}

\section{Spectral~Theory}\label{section3}
\unskip
\subsection{Cash Flows as Sections of the Associated Vector~Bundle} \label{cfbundle}
\begin{defi}[\textbf{Cash Flow Bundle}]
By choosing the fiber $V:=\mathbb{R}^{[0, +\infty[}$
and the representation $\rho:G\rightarrow \text{GL}(V)$ induced by
the gauge transform definition, and~therefore satisfying the
homomorphism relation $\rho(g_1\ast g_2)=\rho(g_1)\rho(g_2)$, we
obtain the associated vector bundle $\mathcal{V}$. Its sections
represent cash flow streams---expressed in terms of the deflators---generated by portfolios of the base assets. If~$v=(v^x_t)_{(t,x)\in M}$ is the \textit{deterministic} cash flow
stream, then its value at time $t$ is equal~to:
\begin{itemize}
\item the deterministic quantity $v_t^x$, if~the value is measured in terms of
the deflator $D_t^x$;
\item the stochastic quantity $v^x_tD^x_t$, if~the value is measured in terms of the
num\'{e}raire (e.g., the cash account for the choice
$D_t^j:=\hat{S}_t^j$ for all $j=1,\dots,N$).
\end{itemize}

 The bundle $\mathcal{V}$ over the time-nominals-space $M=[0,T]\times\mathfrak{X}$
is called the \textbf{Cash Flow Bundle}.
\end{defi}
\noindent In the general theory of
principal fibre bundles, gauge transforms are bundle automorphisms
preserving the group action and equal to the identity on the base
space. Gauge transforms of $\mathcal{B}$ are naturally isomorphic to
the sections of the bundle $\mathcal{B}$ (See Theorem 3.2.2 in~\cite{Bl81}). Since $G$ is Abelian, right multiplications are
gauge transforms. Hence, there is a bijective correspondence between
gauge transforms and cash flow intensities admitting an inverse. This
justifies the terminology introduced in Definition
\ref{gaugeTransforms2}.
\subsection{The Connection Laplacian associated with the Market~Model}
The connection $\chi$ on the market principal fibre bundle $\mathcal{B}$ defined in Theorem \ref{Ilinski} induces a covariant differentiation $\nabla^{\mathcal{V}}$ on the associated vector bundle $\mathcal{V}$, with~the same interpretation for the corresponding parallel transport as that in Theorem \ref{Ilinski} for the principal fibre bundle, i.e.,~portfolio rebalancing along the asset nominal dimensions and discounting along the time dimension. More exactly, we have the following proposition.
\begin{proposition}\label{nabla}
Let us extend the coordinate vector $x\in\mathbb{R}^N$ with a $0$th component given by the time $t$.
Let $X=\sum_{j=0}^NX_j\frac{\partial}{\partial x_j}$ be a vector field over $M$ and $f=(f_s)_s$ a section of the cash flow bundle $\mathcal{V}$. Then
\begin{equation}
\nabla^{\mathcal{V}}_Xf_t=\sum_{j=0}^N\left(\frac{\partial f_t}{\partial x_j}+K_jf_t\right)X_j,
\end{equation}
where
\begin{equation}
\begin{split}
K_0(x)&=-r_t^x\\
K_j(x)&=\frac{D^j_t}{D_t^x}\quad(1\le j\le N).
\end{split}
\end{equation}
\end{proposition}
\begin{proof} The construction of a covariant differentiation on the associated vector bundle starting from a connection on a principal fibre bundle is a generic procedure in differential geometry. The~connection $\chi$ is a Lie algebra $\mathfrak{g}=\mathbb{R}^{[0,+\infty[}$ valued $1$-form on $\mathcal{B}$, and we can decompose the connection as $\chi(x,g)=gK(x)$, where $K(x):=\sum_{j=0}^NK_j(x)dx_j$. The~differential map $T_e\rho:\mathfrak{g}\rightarrow\mathcal{L}(\mathbb{R}^{[0,+\infty[})$ of the representation $\rho:G\rightarrow\text{GL}(\mathbb{R}^{[0,+\infty[})$ maps elements of the Lie algebra on endomorphisms for the bundle $\mathcal{V}$.
Given a local cash flow section $f_t=\int_0^{+\infty}ds\,f_s\delta_{s-t}$, in~$\mathcal{V}|_U$ and a local vector field $X$ in $TM|_U$, the connection $\nabla^{\mathcal{V}}$ has the following local representation:
\begin{equation}
\nabla^{\mathcal{V}}_Xf_t=\int_0^{+\infty}ds(df_s(X).v_s+f_s\omega(X).v_s),
\end{equation}
where $v_s:=\delta_{s-t}$ and $\omega$ is an element of $T^*U|_U\bigotimes \mathcal{L}(V|_U)$, i.e.,~an endomorphism valued $1-$form defined as follows:
\begin{equation}
\omega(x)(X):=(T_e\rho.\chi(x,e)).X=\left.\frac{d}{d\varepsilon}\right|_{\varepsilon=0}\rho\left(\exp(\varepsilon \chi(x,e).X)e\right).
\end{equation}

Since the derivative of the exponential map is the identity and
\begin{equation}
\rho(\pi)=\pi*\cdot\in\text{GL}(\mathcal{V}_x)\Rightarrow T_e\rho.t=t*\cdot\in\mathcal{L}(\mathcal{V}_x),
\end{equation}
\noindent it follows that
\begin{equation}
\omega(x)= \chi(x,e)*\cdot=K(x)\otimes\delta*\cdot,
\end{equation}
and therefore,
\begin{equation}
\begin{split}
\nabla^{\mathcal{V}}_Xf_t&=\int_0^{+\infty}ds\left[df_s(X)v_s+f_sK.X\delta*v_s\right]\\
&=\int_0^{+\infty}ds\left[df_s(X)+f_sK.X\right]v_s\\
&=df_t(X)+f_tK.X\\
&=\sum_{j=0}^N\left(\frac{\partial f_t}{\partial x_j}+K_jf_t\right)X_j.
\end{split}
\end{equation}
\end{proof}

\begin{proposition}
The curvature of the connection $\nabla^{\mathcal{V}}$ is
\begin{equation}\label{curvature_equiv}
R^{\mathcal{V}}(X,Y):=\nabla^{\mathcal{V}}_X\nabla^{\mathcal{V}}_Y-\nabla^{\mathcal{V}}_Y\nabla^{\mathcal{V}}_X-\nabla^{\mathcal{V}}_{[X,Y]}=[p]\circ(R(X^*,Y^*,e)\,*\, \cdot)\circ [p]^{-1},
\end{equation}
where $R$ is the curvature on the principal fibre bundle $\mathcal{B}$, $X^*,Y^*\in T_p\mathcal{B}$ the horizontal lifts of $X,Y\in T_{(t,x)}M$, and
\begin{equation}
\begin{split}
[p]:V=\mathbb{R}^{[0, +\infty[}&\longrightarrow\mathcal{V}_{(t,x)}:=\mathcal{B}_{(t,x)}\times_{G}V\\
v&\longmapsto [p](v)=[p,v]
\end{split}
\end{equation}
\noindent is the fibre isomorphism between $\mathcal{B}$ and $\mathcal{V}$. In~particular, the curvature on the principal fibre bundle vanishes if and only if the curvature on the associated vector bundle vanishes.
\end{proposition}
\begin{proof}
Equation~(\ref{curvature_equiv}) follows from the definition of the curvature on a vector bundle and utilizes  Satz 3.21 in~\cite{Ba14} with $\rho(\pi)=\pi*\,\cdot\,$.
\end{proof}

 We now continue by introducing the connection Laplacian on an appropriate Hilbert space.
\begin{defi} The space of the sections of the cash flow bundle can be turned into a scalar product space by introducing,
for stochastic sections $f=f(t, x, \omega)=(f_s(t, x, \omega))_{s\in[0,+\infty[}$ and $g=g(t, x, \omega)=(g_s(t, x, \omega))_{s\in[0,+\infty[}$, the following:
\begin{align}
(f,g):=\int_{\Omega}dP\int_Xd^Nx\int_0^{+\infty}dt\left<f,g\right>(t, x, \omega) &=\mathbb{E}_0\left[(f,g)_{L^2(M,\mathbb{R}^{[0,+\infty[})}\right] \nonumber \\
&=(f,g)_{L^2(\Omega,\mathcal{V},\mathcal{A}_0,dP)},
\end{align}
where
\begin{equation*}
\left<f,g\right>(x, t, \omega):=\int_0^{+\infty}dsf_s(t, x, \omega)g_s(t,x, \omega).
\end{equation*}

The Hilbert space of integrable sections reads as follows:
\begin{align}
\mathcal{H}&:=L^2(\Omega,\mathcal{V},\mathcal{A}_0,dP) \nonumber \\
&=\left\{\left.f=f(t,x,\omega)=(f_s(t,x,\omega))_{s\in[0,+\infty[}\right|\,(f,f)_{L^2(\Omega,\mathcal{V},\mathcal{A}_0,dP)}<+\infty\right\}.
\end{align}
\end{defi}

 When considering the connection Laplacian, there are two standard choices for a local elliptic boundary condition which guarantees~self-adjointness:
\begin{itemize}
\item \textbf{Dirichlet boundary condition:} 

\begin{equation}
B_D(f):=f|_{\partial M}.
\end{equation}
\item \textbf{Neumann boundary condition:}
\begin{equation}
B_N(f):=(\nabla^{\mathcal{V}}_{\nu}f)|_{\partial M},
\end{equation}
\noindent where $\nu$ denotes the normal unit vector field to $\partial M$.
\end{itemize}

By considering the $\omega$ a parameter dependence, we can apply a standard result functional analysis to obtain the following proposition.
\begin{proposition}\label{spec}
The connection Laplacian $\Delta^{\mathcal{V}}:={\nabla^{\mathcal{V}}}^*\nabla^{\mathcal{V}}$ with a domain definition given by the Neumann boundary condition,
\begin{equation}
\text{dom}\left(\Delta^{\mathcal{V}}_{B_N}\right):=\left\{f\in\left.\mathcal{H}\right|\,f(\omega,\cdot,\cdot)\in H^2(M,\mathbb{R}^{[0,+\infty[}),\,B_N(f(\omega,\cdot,\cdot))=0\;\forall\, \omega\in\Omega\right\},
\end{equation}
is a self-adjoint operator on $\mathcal{H}$. Its spectrum consists in the disjoint union of the discrete spectrum (eigenvalues) and the continuous spectrum (approximate eigenvalues) lying in $[0,+\infty[$:
\begin{equation}
\text{spec}(\Delta^{{\mathcal V}}_{B_N})=\text{spec}_d(\Delta^{{\mathcal V}}_{B_N})\;\dot{\cup} \text{spec}_c(\Delta^{{\mathcal V}}_{B_N}).
\end{equation}

 If $M$ is compact, for~example, by setting $M:=[0,T]\times\mathfrak{X}$, $\mathfrak{X}\subset\mathbb{R}^N$ compact and $T<+\infty$, then the continuous spectrum is empty and the eigenvalues can be ordered in a monotone increasing sequence converging to $+\infty$.
\end{proposition}
\begin{remark}
When we choose $M=[0,T]\times\mathfrak{X}$ with $T<+\infty$, we have to adapt the construction of the principal fibre bundle and the associated vector bundle accordingly. Note that the structure group of $\mathcal{B}$ and its Lie Algebra remain $G$ and $\mathbb{R}^{[0, +\infty[}$, respectively, and~the fibre of $\mathcal{V}$ is still $\mathbb{R}^{[0, +\infty[}$. Only the integration over the time dimension in the base space $M$ is performed until $T$.
\end{remark}
\begin{remark}\label{spec_rem}
For a fixed $\omega\in\Omega$, the definition domain  of $\Delta^{\mathcal{V}}_{B_N}$ is a subset of the Sobolev space $H^2(M,\mathbb{R}^{[0,+\infty[})$. If~$M$ is compact, then the eigenvectors of $\Delta^{\mathcal{V}}_{B_N}$ lie in $C^{\infty}(M,\mathbb{R}^{[0,+\infty[})$ and satisfy the Neumann boundary condition. Proposition \ref{spec} follows from standard elliptic spectral theory by means of an integration over $\Omega$.
\end{remark}

 The spectrum of the connection Laplacian under the Neumann boundary condition contains information about arbitrage possibilities in the market.
\begin{theorem}\label{spec_NFLVR}
The market model satisfies the NFLVR condition if and only if $0\in \text{spec}_d\left(\Delta^{{\mathcal V}}_{B_N}\right)$. The~harmonic sections parameterize the Radon--Nikodym derivative for the change of measure from the statistical to the risk-neutral measures.
\end{theorem}
\begin{proof}
The spectrum of the Laplacian under Neumann boundary conditions contains $0$ if and only if there exists a section $f$ such that
\begin{equation}
\nabla^{{\mathcal V}} f=0.
\end{equation}

According to Proposition \ref{nabla}, this is equivalent to
\begin{equation}
\frac{\partial f_t}{\partial x_j}+K_jf_t=0,
\end{equation}
for all $j=0,1,\dots,N$. This means that for $j=0$
\begin{equation}\label{eqR}
{\mathcal D}\log(f_t(x))-r_t^x=0,
\end{equation}
and for~$j=1,\dots,N$,
\begin{equation}\label{eqD}
\frac{\partial \log(f_t(x))}{\partial x_j}=-\frac{D_t^j}{D_t^x},
\end{equation}
\noindent for all $x\in \mathfrak{X}$. Equation~(\ref{eqD}) becomes
\begin{equation}\label{eqDbis}
\begin{split}
&\frac{\partial \log(f_t(x))}{\partial x_j}=-\frac{\partial \log(D_t^x)}{\partial x_j}\\
&\\
&\frac{\partial \log(f_t(x))D_t^x}{\partial x_j}=0\\
&\\
&\log(f_t(x))D_t^x)\equiv C_t\\
&\\
&f_t(x)D_t^x\equiv\exp(C_t),
\end{split}
\end{equation}
\noindent for a process $(C_t)_{t\in[0,+\infty[}$. Therefore, the~positive process $\beta=(\beta_t)_{t\in[0,+\infty[}:=(\exp(-C_t))$ $_{t\in[0,+\infty[}$ satisfies
\begin{equation}\label{f}
f_t(x)=\frac{1}{\beta_tD_t^x},
\end{equation}
which, when inserted into Equation~(\ref{eqR}), leads to
\begin{equation}
\mathcal{D}\log(\beta_tD_t^x)+r_t^x=0,
\end{equation}
\noindent for all $t$ and $x$.

For fixed $\omega\in\Omega$, the Laplace operator has an elliptic symbol and by Weyl's theorem, any harmonic $f=f(\omega,t,x)$ is a smooth function of $(t,x)$. In~particular, any path of $f$ is c\`{a}dl\`{a}g with bounded variation, and hence $(f_t)_t$ is a semimartingale. According to \mbox{Equation~(\ref{f})}, $(D_t)_t$ being a semimartingale, it follows that $(\beta_t)_t$ is a semimartingale as well. Based on \mbox{Theorem \ref{holonomy}} this is equivalent to the NFLVR condition.
\end{proof}
\begin{remark}
Note that if $f=f(\omega,t,x)\equiv f(\omega)$, and~at least one of the components of $r$ or $D$ does not vanish, then $f=0$, $0\notin {\rm spec}\left(\Delta^{{\mathcal V}}_{B_N}\right)$, confirming and extending Remark \ref{rem_holonomy}.
\end{remark}

\begin{remark}Any harmonic $f=f_t(x)$ defines a risk-neutral measure by means of the Radon--Nikodym derivative:
\begin{equation}\label{Rad}
\frac{dP^*}{dP}=\frac{\beta_t}{\beta_0}=\frac{D_0^x}{D_t^x}\frac{f_0(x)}{f_t(x)},
\end{equation}
\noindent which does not depend on $x$.
\end{remark}
From Formula (\ref{Rad}), we derive the following corollary. 

\begin{corollary}\label{spec_Complete}
The market model is complete if and only if $0\in {\rm spec}(\Delta^{\mathcal{V}})_{B_N}$ is an eigenvalue with simple multiplicity.
\end{corollary}

\begin{remark}
The situation for the Dirichlet boundary condition is similar. The~proposition and remark analogous to Proposition \ref{spec} and Remark \ref{spec_rem} hold true.
But because of the unique continuation property for elliptic operators, $0$ never lies in ${\rm spec}_d\left(\Delta^{\mathcal{V}}_{B_D}\right)$ whether the NFLVR property is satisfied or not.
\end{remark}

\subsection{Arbitrage~Bubbles}
\begin{defi}[\textbf{Spectral Lower Bound}]
The highest spectral lower bound of the connection Laplacian on the cash flow bundle $\mathcal{V}$ is given by the following:
\begin{equation}
\lambda_0:=\inf_{\substack{\varphi\in C^{\infty}(M,\mathcal{V})\\\varphi\neq 0\\B_N(\varphi)=0}}\frac{(\nabla^{\mathcal{V}}\varphi, \nabla^{\mathcal{V}}\varphi)_{\mathcal{H}}}{(\varphi,\varphi)_{\mathcal{H}}},
\end{equation}
and it is assumed on the subspace
\begin{equation}
E_{\lambda_0}:=\left\{\varphi\left|\,\varphi\in C^{\infty}(M,\mathcal{V})\cap\mathcal{H}, \, B_N(\varphi)=0, (\nabla^{\mathcal{V}}\varphi, \, \nabla^{\mathcal{V}}\varphi)_{\mathcal{H}}\ge \lambda_0 (\varphi,\varphi)_{\mathcal{H}}\right.\right\}.
\end{equation}
The space
\begin{equation}
\mathcal{K}_{\lambda_0}:=\{\varphi\in E_{\lambda_0}\left|\,\varphi\ge0,\,\mathbb{E}[\varphi]=1\right.\}
\end{equation}
contains all candidates for the Radon--Nikodym derivative
\begin{equation}\label{risk_neutral_measure}
\frac{dP^*}{dP}=\varphi,
\end{equation}
for a probability measure $P^*$ absolutely continuous with respect to the statistical measure $P$.
\end{defi}

 By reformulating Theorem \ref{spec_NFLVR} and Corollary \ref{spec_Complete}, we obtain following statement.
\begin{proposition}
The market model satisfies the NFLVR condition if and only if $\lambda_0=0$. Therefore, there exist risk-neutral probability measures defined by (\ref{risk_neutral_measure}) with the corresponding $\varphi\in\mathcal{K}_{\lambda_0}$ such that $(D_t)_{t\in[0,T]}$ is a vector valued martingale with respect to $P^*$, i.e.,
\begin{equation}
\mathbb{E}^*_t[D_s]
=D_t\qquad\text{ for all }s\ge t \text{ in }[0,T].
\end{equation}
The market is complete if and only if $\lambda_0=0$ and $\dim E_{0}=1$.
\end{proposition}

 For arbitrage markets we have $\lambda_0>0$, and there exist no risk-neutral probability measures. Nevertheless, it is possible to define a fundamental value, although not in a unique way.
\begin{defi}[\textbf{Basic Assets' Arbitrage Fundamental Prices and Bubbles}]\label{arbitrage_bubble}
Let $(C_t)_{t\in[0,T]}$ be the ${\mathbb R}^N$ cash flow stream stochastic process associated with the $N$ assets of the market model with a given spectral lower bound $\lambda_0$ and the Radon--Nikodym subspace $\mathcal{K}_{\lambda_0}$.
For a given choice of $\varphi\in\mathcal{K}_{\lambda_0}$, the approximated fundamental value of the assets with a stochastic $\mathbb{R}^N$-valued price process $(S_t)_{t\in[0,T]}$ is defined as follows:
\begin{equation}
\boxed{
S_t^{*,\varphi}:=\mathbb{E}_t\left[\varphi\left(\int_t^\tau dC_u\,\exp\left(-\int_t^ur_s^0ds\right)+S_\tau\exp\left(-\int_t^\tau r_s^0ds\right)\,1_{\{\tau<+\infty\}}\right)\right]\,1_{\{t<\tau \}},
}
\end{equation}
\noindent where $\tau$ denotes the maturity time of all risky assets in the market model, and the~approximated bubble is defined as follows:
\begin{equation}
\boxed{
B_t^{\varphi}:=S_t-S_t^{*,\varphi}.
}
\end{equation}
The fundamental price vector for the assets and their asset bubble prices are defined as:
\begin{equation}\label{phi_0}
\boxed{
\begin{split}
S^*_t&:=\tilde{S}_t^{*,\varphi_0}\\
B_t&:=B_t^{\varphi_0}\\
\varphi_0&:=\arg\min_{\varphi\in \mathcal{K}_{\lambda_0}}\mathbb{E}_0\left[\int_0^Tds\,|B_s^{\varphi}|^2\right].
\end{split}
}
\end{equation}
The probability measure $P^*$ with the Radon--Nikodym derivative
\begin{equation}
\frac{dP^*}{dP}=\varphi_0
\end{equation}
is termed \textbf{minimal arbitrage measure}.
\end{defi}

\begin{proposition}
The assets' fundamental values can be expressed as conditional expectation with respect to the minimal arbitrage measure using the following formula:
\begin{equation}
\boxed{
S^*_t:=\mathbb{E}_t^*\left[\int_t^\tau dC_u\,\exp\left(-\int_t^ur_s^0ds\right)+S_\tau\exp\left(-\int_t^\tau r_s^0ds\right)\,1_{\{\tau<+\infty\}}\right]\,1_{\{t<\tau \}}.
}
\end{equation}
\end{proposition}

\begin{defi}[\textbf{Scalar curvature}] The market \textbf{integral scalar curvature} at time $t$ for the portfolio $x$ is defined as follows:
\begin{equation}
\mathcal{K}(t,x):=\mathcal{D}\log D^x_t+r^x_t.
\end{equation}
A strategy $x=(x_t)_{t\in[0,T]}$ is a \textbf{free lunch / no-arbitrage / rip-off strategy} if and only if
\begin{equation}
\boxed{
\mathcal{K}(t,x_t)\left\{
                    \begin{array}{ll}
                      >0 & \hbox{(free lunch)} \\
                      =0 & \hbox{(no-arbitrage)} \\
                      <0 & \hbox{(rip off)}
                    \end{array}
                  \right.
\quad \text{ for all}\quad t\in[0,T].
}
\end{equation}
The \textbf{vector valued integral curvature}  of the portfolio is defined as the vector of integral scalar curvatures for the portfolio single asset components:
\begin{equation}\label{vvcurvature}
\overrightarrow{\mathcal{K}}(t,x):=\sum_{j=1}^N\left(\mathcal{D}\log D^{x^je_j}_t+r^{x^je_j}_t\right)e_j.
\end{equation}
\end{defi}
\begin{remark}
The curvature $R$ defined in (\ref{defCurv}) 
  can be written as follows:
\begin{equation}
R(t,x,g)=g dt\wedge d_x\mathcal{K}(t,x),
\end{equation}
\noindent therefore justifying the nomenclature of the integral scalar curvature $\mathcal{K}$.
\end{remark}
We can now extend Jarrow--Protter--Shimbo's result in~\cite{JPS10} to obtain the following bubble decomposition theorem.
\begin{theorem}[\textbf{Bubble decomposition}]\label{arbitrage_bubble_contingent}
Let $\tau$ denote the maturity time of all risky assets in the market model. $S_t$ admits a unique (up to the $P$-evanescent set) decomposition
into a sum of fundamental and bubble values:
\begin{equation}
S_t=S^*_t+B_t,
\end{equation}
\noindent where $(B_t)_{t\in[0,T]}$ is a c\`{a}dl\`{a}g process satisfying
\begin{equation}\label{arbitrage_bubble_equation}
\boxed{
\begin{split}
&B_t=S_t+\\
&\qquad-\mathbb{E}_t\left[\varphi_0\left\{\int_t^\tau dC_u\exp\left(-\int_t^udsr^0_s\right)+\right.\right.\\
&\qquad\qquad\qquad\quad\left.\left.+S_t\exp\left(\int_t^\tau ds(\overrightarrow{\mathcal{K}}(s,e)-r^0_s)\right)1_{\{\tau<+\infty\}}\right\}\right]1_{\{t<\tau\}},
\end{split}
}
\end{equation}
\noindent where $e=[1,\dots,1]^\dagger$ or, equivalently, for~all $j=1,\dots,N$
\begin{equation}\label{bubble_value}
\boxed{
B_t^j=S_t^j-\mathbb{E}_t^*\left[\int_t^\tau dC_u^j\,\exp\left(-\int_t^uds\, r^0_s\right)+\exp\left(-\int_t^\tau ds\,r_s^0\right)S_{\tau}^j1_{\{\tau<+\infty\}}\right]\,1_{\{t<\tau\}}.
}
\end{equation}
If all asset maturities are finite, i.e.,~$\tau=T<+\infty$, then
\begin{equation}
\boxed{
\begin{split}
\overrightarrow{\mathcal{K}}_\cdot>r_\cdot^0\text{ and }C_\cdot>0&\Longrightarrow B_t\uparrow0^-\,(t\rightarrow T^-)\\
\overrightarrow{\mathcal{K}}_\cdot<r_\cdot^0\text{ and }C_\cdot<0&\Longrightarrow  B_t\downarrow0^+\,(t\rightarrow T^-)\\
\overrightarrow{\mathcal{K}}_\cdot=r_\cdot^0\text{ and }C_\cdot=0&\Longrightarrow  B_t\equiv0,
\end{split}
}
\end{equation}
\noindent where the inequalities and the limits are meant componentwise.
\end{theorem}

\begin{proof}
By developing the expression for the bubbles' values and utilizing the definition of deflators for all $t\ge0$
\begin{equation}
D_t=\exp\left(-\int_0^t ds\,r_s^0\right)S_t,
\end{equation}
\noindent since the curvature is the instantaneous asset portfolio log return (see~\cite{FaTa20})
\begin{equation}
D_\tau=D_t\exp\left(\int_t^\tau  ds\,\left(\mathcal{D}\log D_s+r_s\right)\right),
\end{equation}
\noindent we obtain
\begin{equation}
S_\tau=S_t\exp\left(\int_t^\tau  ds\,\left(\mathcal{D}\log D_s+r_s-r_s^0\right)\right),
\end{equation}
\noindent and hence,
\begin{equation}
\begin{split}
B_t&=S_t-\mathbb{E}_t\left[\varphi_0\left\{\int_t^\tau dC_u\,\exp\left(-\int_t^u ds\, r_s^0\right)+\right.\right.\\
&\qquad\qquad\left.\left.+S_{\tau}\exp\left(-\int_t^\tau ds\,  r_s^0\right)\,1_{\{\tau<+\infty\}}\right\}\right]\,1_{\{t<\tau\}}=\\
&=S_t-\mathbb{E}_t\left[\varphi_0\left\{\int_t^\tau dC_u\,\exp\left(-\int_t^u ds\, r^0_s\right)+\right.\right.\\
&\qquad\qquad+\left.\left.S_t\exp\left(\int_t^\tau  ds\,(\overrightarrow{\mathcal{K}}(s,e)-r^0_s)\right)1_{\{\tau<+\infty\}}\right\}\right]\,1_{\{t<\tau\}},
\end{split}
\end{equation}
\noindent which is (\ref{arbitrage_bubble_equation}), and for~finite $T$, it becomes\vspace{-8pt}
\begin{equation}\label{bubble_eq}
\begin{split}
B_t&=S_t-\mathbb{E}_t\left[\varphi_0\left(\int_t^T\, dC_u\,\exp\left(-\int_t^u ds\, r^0_s\right)+\right.\right.\\
&\qquad\qquad\left.\left.+S_t\exp\left(\int_t^T ds\,(\overrightarrow{\mathcal{K}}(s,e)-r^0_s)\right)\right)\right]=\\
&=S_t-S_t\underbrace{\mathbb{E}_t\left[\varphi_0\exp\left(\int_t^T ds\,(\overrightarrow{\mathcal{K}}(s,e)-r^0_s)\right)\right]}_{=:A_1(t,T)}+\\
&\qquad\qquad-\underbrace{\mathbb{E}_t\left[\varphi_0\left(\int_t^T\, dC_u\,\exp\left(-\int_t^u ds\, r_s^0\right)\right)\right]}_{=:A_2(t,T)}.
\end{split}
\end{equation}
Now, since $\mathbb{E}_0[\varphi_0]=1$, we observe that\vspace{-10pt}
\begin{equation}
\begin{split}
\overrightarrow{\mathcal{K}}_\cdot-r_\cdot^0>0,\,C_\cdot>0\Longrightarrow &A_1(t,T)>0, A_2(t,T)>0,\\
&  A_1(t,T)\downarrow1^+, A_2(t,T)\downarrow0^+\quad(t\rightarrow T^-)\\
\overrightarrow{\mathcal{K}}_\cdot-r_\cdot^0<0,\,C_\cdot<0\Longrightarrow &A_1(t,T)>0, A_2(t,T)<0,\\
& A_1(t,T)\uparrow1^+, A_2(t,T)\uparrow0^+\quad(t\rightarrow T^-),
\end{split}
\end{equation}
\noindent  from which we conclude according to (\ref{bubble_eq}) that
\begin{equation}
\begin{split}
\overrightarrow{\mathcal{K}}_t-r_t^0>0\text{ for }t\in[0,T],\, C_\cdot>0&\Longrightarrow B_t\uparrow0^-\;(t\rightarrow T^-)\\
\overrightarrow{\mathcal{K}}_t-r_t^0<0\text{ for }t\in[0,T],\, C_\cdot<0&\Longrightarrow  B_t\downarrow0^+\;(t\rightarrow T^-),
\end{split}
\end{equation}
\noindent and hence,
\begin{equation}
\overrightarrow{\mathcal{K}}_t-r_t\equiv 0\text{ and }C=0\Longrightarrow  B_t\equiv0.
\end{equation}
Inserting the definition (\ref{vvcurvature}) of $\overrightarrow{\mathcal{K}}$ into (\ref{arbitrage_bubble_equation}) leads to (\ref{bubble_value}). The~proof is now~completed.
\end{proof}

  We can now extend Jarrow--Protter--Shimbo's result in~\cite{JPS10} to obtain the following bubble classification theorem.
\begin{theorem}[\textbf{Bubble types}]
Let $T=+\infty$, and~denote as $\tau$ the maturity time of all risky assets in the market model. If~there exists a non-trivial bubble $B_t^j$ in an asset's price for $j=1,\dots,N$, then there exists at least one probability measure $P^*$ equivalent to $P$, for~which we have three and only three~possibilities:
\begin{itemize}
\item \textbf{Type1: } $B_t^j$ is local super- or submartingale with respect to both $P$ and $P^*$, if~$P[\tau=+\infty]>0$;
\item \textbf{Type2: } $B_t^j$ is local super- or submartingale  with respect to both $P$ and $P^*$, but~not uniformly integrable super- or submartingale, if~$B_t^j$ is unbounded but with $P[\tau<+\infty]=1$;
\item \textbf{Type3: } $B_t^j$ is a strict local super- or sub- $P$- and $P^*$-martingale, if~$\tau$ is a bounded stopping time.
\end{itemize}
\end{theorem}
\begin{proof}
This theorem is a direct consequence of a local application of Equation~(\ref{bubble_eq}), applied on time subintervals of $[0,T]$ on~which the signs of ${\mathcal K}(t,e_j)+r_t^0$ and of $C_t$ remain constant.
On the subintervals with non-negative ${\mathcal K}(t,e_j)$ and $C_t$, the~bubble price $B_t^j$ is a sub-martingale for both $P$ and $P^*$, because~for $t\le s$ that is near enough, $B_t\le B_s$ holds true;~thus,
\begin{equation}
B_t\le {\mathbb E}_t[B_s]\quad\text{and}\quad B_t\le {\mathbb E}_t^{*}[B_s].
\end{equation}

 On the subintervals with non-positive ${\mathcal K}(t,e_j)$ and $C_t$, the~bubble price $B_t^j$ is a super-martingale for both $P$ and $P^*$, because~for $t\le s$ that is near enough, $B_t\ge B_s$ holds true;~thus,
\begin{equation}
B_t\ge {\mathbb E}_t[B_s]\quad\text{and}\quad B_t\ge {\mathbb E}_t^{*}[B_s].
\end{equation}

 In the case of type $1$, there is a non-evanescent set of elementary events for which the maturity time of the assets is not finite, and~without further information, we do not know whether the stochastic integral until infinite converges. In~the case of type $2$, the set of elementary events for which the maturity time of the assets $\tau$ is not finite vanishes a.s., but~if $B_t^j$ is unbounded, we do not know if the stochastic integral over $[t, \tau]$ converges uniformly in $\tau$. In~the case of type $3$, the stochastic integral converges. The~proof is~completed.
\end{proof}
\begin{remark} Note that if the NFLVR is satisfied, then the curvature vanishes, and so do the bubbles of type $3$, these being trivial martingales for both the statistical and the risk-neutral probability measures.
\end{remark}

\begin{defi}[\textbf{Contingent Claim's Arbitrage Fundamental Price and Bubble}]
Let us consider in the context of Definition (\ref{arbitrage_bubble}) a European option given by the contingent claim  with a unique payoff $H(S_T)$ at time $T$ for an appropriate real-valued function $H$ of $N$ real variables. The fundamental price~of the contingent claim and its corresponding arbitrage bubble is defined in the case of base assets paying no dividends:
\begin{equation}\label{option_price}
\boxed{
\begin{split}
V^*_t(H)&:=\mathbb{E}_t\left[\varphi_0\exp\left(-\int_t^Tr_s^0ds\right)H(S_T)\,1_{\{T<+\infty\}}\right]1_{\{t<T\}}\\
&=\mathbb{E}^*\left[\exp\left(-\int_t^Tr_s^0ds\right)H(S_T)\,1_{\{T<+\infty\}}\right]1_{\{t<T\}}\\
B_t(H)&:=V_t(H)-V^*_t(H),
\end{split}
}
\end{equation}
\noindent where $\varphi_0$ is the minimizer for the basic assets bubbled defined in (\ref{phi_0}), $P^*$ the minimal arbitrage measure, and $(V_t(H))_{t\in[0,T]}$ is the price process of the European~option.\par
In the case of base assets paying dividends, the definition becomes the following:
\begin{equation}\label{option_price_div}
\boxed{
\begin{split}
V^*_t(H)&:=\mathbb{E}_t\left[\varphi_0\exp\left(-\int_t^Tr_s^0ds\right)H\left(S_T\exp\left(\frac{C_T}{S_T}(T-t)\right)\right)\,1_{\{T<+\infty\}}\right]1_{\{t<T\}}\\
&=\mathbb{E}^*\left[\exp\left(-\int_t^Tr_s^0ds\right)H\left(S_T\exp\left(\frac{C_T}{S_T}(T-t)\right)\right)\,1_{\{T<+\infty\}}\right]1_{\{t<T\}}\\
B_t(H)&:=V_t(H)-V^*_t(H),
\end{split}
}
\end{equation}
\noindent where $\frac{C_t^j}{S_t^j}$ is the instantaneous dividend rate for the $j$-th asset.
\end{defi}

\begin{remark}
If the market is complete, then $\lambda_0=0$ and $\mathcal{K}_{\lambda_0}=\{\varphi_0\}$, where $\varphi_0$ is the Radon--Nikodym derivative of the unique risk-neutral probability measure with respect to the statistical probability measure. The~definitions in (\ref{arbitrage_bubble}) and in (\ref{arbitrage_bubble_contingent}) for the complete market coincide with the definitions of fundamental value and asset bubble price for both base asset and contingent claim as introduced by Jarrow, Protter, and Shimbo in~\cite{JPS10}, proving that they are a natural extension to markets allowing for arbitrage opportunities.
\end{remark}
\begin{remark}
We see that the fundamental price of an asset defined via minimal arbitrage measure does share common characteristics with the real world pricing in the benchmark approach by Platen and Heath (see chapters 9 and 10 of~\cite{HePl06}).
\end{remark}
We now prove the put-call parity of fundamental prices.
\begin{proposition}[\textbf{Put-Call Parity for Fundamental Prices}]
Let us consider the market model with $N=1$ for the base assets (i.e., cash and one risky asset). Then, the~fundamental price~processes:
\begin{itemize}
\item $(C^*_t)_{t\in[0,T]}$ of a call option on $(S_t)_{t\in[0,T]}$ with strike price $K>0$ at time $T$;
\item $(P^*_t)_{t\in[0,T]}$ of a put option on $(S_t)_{t\in[0,T]}$ with strike price $K>0$ at time $T$;
\item $(F^*_t)_{t\in[0,T]}$ of a forward on $(S_t)_{t\in[0,T]}$ with forward price $K>0$ at time $T$;
\end{itemize}
satisfy the put-call-parity relation if the base asset pays no dividends.
\begin{equation}\label{put_call_parity}
\boxed{
C^*_t-P^*_t=F^*_t.
}
\end{equation}
\end{proposition}

\begin{proof}
For a positive strike price $K>0$, we can decompose the forward payoff as:
\begin{equation}\label{cf_dec}
S_T-K=(S_T-K)^+-(K-S_T)^+.
\end{equation}

 Formula (\ref{option_price}), in the case of no dividends ($C\equiv0$), reads as:
\begin{equation}
V^*_t(H)=\mathbb{E}_t\left[\varphi_0\exp\left(-\int_t^Tr_s^0ds\right)H(S_T)\,1_{\{T<+\infty\}}\right]1_{\{t<T\}},
\end{equation}
\noindent which, applied to (\ref{cf_dec}), leads to
\begin{equation}
V^*_t((S_T-K)^+)-V^*_t((K-S_T)^+)=V^*_t((S_T-K)),
\end{equation}
\noindent which is Equation~(\ref{put_call_parity}). The~proof is~completed.\\
\end{proof}
\begin{remark} The put-call parity for market prices may be violated even under the NFLVR assumption (see~\cite{Pr13}).
\end{remark}

  Finally, a~short direct computation shows the corollary below.
\begin{corollary}
The bubble discounted values for the base assets in Definition \ref{arbitrage_bubble}  and for the contingent claim on the base assets paying dividends in Definition \ref{arbitrage_bubble_contingent}
\begin{equation}
\widehat{B}_t:=\exp\left(-\int_0^tr_s^0ds\right)B_t,\qquad \widehat{B}(H)_t:=\exp\left(-\int_0^tr_s^0ds\right)B(H)_t
\end{equation}
\noindent satisfy the equalities
\begin{equation}
\boxed{
\begin{split}
&\widehat{B}_t^j =D_t^j-\left(\mathbb{E}_t^*\left[D_{\tau}^j 1_{\{\tau<+\infty\}}\right]+\mathbb{E}_t^*\left[\widehat{C}_\tau^j 1_{\{\tau<+\infty\}}\right]-\widehat{C}_t^j\right) 1_{\{t<\tau\}}\\
&\widehat{B}_t(H)=\widehat{V}_t(H)-\mathbb{E}_t^*\left[\widehat{H}\left(S_T\exp\left(\frac{C_T}{S_T}(T-t)\right)\right)1_{\{T<+\infty\}}\right]1_{\{t<T\}}.
\end{split}
}
\end{equation}
\noindent where
\begin{align}
\widehat{C}_t^j&:=\exp\left(-\int_0^tr_s^0ds\right)C_t^j\\
\widehat{H}&:=\exp\left(-\int_0^Tr_s^0ds\right)H\\
  \widehat{V}_t(H)&:=\exp\left(-\int_0^tr_s^0ds\right)V_t(H)
\end{align}
\noindent are the discounted cash flow for the $j$-th asset, the~discounted contingent claim payoff, and~the discounted value of the derivative.
\end{corollary}

\section{Topological Obstructions to~Arbitrage}\label{section4}
\unskip

\subsection{Topological Obstruction to Arbitrage Induced by the Gauss-Bonnet-Chern~Theorem}
We want to show that the exterior algebra bundle twisted with the cash flow bundle defined in Section \ref{cfbundle} can be given the structure of a Dirac bundle. Then, we apply the version of the Atiyah--Singer index theorem for this bundle, called the Gauss--Bonnet--Chern theorem, which relates the integral of the Euler form (i.e., the Pfaffian of the curvature of the connection) with a topological invariant of the bundle, the~Euler characteristic. First, we recall some basic definitions and examples about Dirac operators.
\begin{defi}\label{DiracBundle}The quadruple $(W,\langle \cdot,\cdot \rangle,\nabla,\gamma)$, where
  \begin{enumerate}
  \item[(i)] $W$ is a complex (real) vector bundle over the oriented Riemannian manifold $(M,g)$ with a
Hermitian (Riemannian) structure $\langle \cdot,\cdot \rangle $;
  \item[(ii)] $\nabla:C^\infty(M,W)\to C^\infty(M,T^*M\otimes W)$ is a connection on $M$;
  \item[(iii)] $\gamma: {\rm Cl}(M,g)\to {\rm Hom}(W)$ is a {\it real} algebra bundle homomorphism from
the Clifford bundle over $M$ to the {\it real} bundle of complex (real) endomorphisms of $W$, i.e.,~$W$ is a bundle of Clifford modules;
 \end{enumerate}
    is said to be a \textbf{Dirac bundle} if~the following conditions are~satisfied:
    \begin{enumerate}
    \item[(iv)] $\gamma(v)^*=-\gamma(v)$, $\forall v\in TM$, i.e.,~the Clifford multiplication
by tangent vectors is fiberwise skew-adjoint with respect to the Hermitian (Riemannian) structure $\langle \cdot,\cdot \rangle $;
    \item[(v)] $\nabla \langle \cdot,\cdot \rangle =0$, i.e.,~the connection is Leibnizian (Riemannian). In~other words, it satisfies the product rule:
\begin{equation}
        d\langle \varphi,\psi \rangle =\langle \nabla\varphi,\psi \rangle +\langle \varphi,\nabla\psi \rangle ,\quad\forall
\varphi,\psi\in C^\infty(M,W);
      \end{equation}
    \item[(vi)] $\nabla\gamma=0$, i.e.,~the connection is a module derivation. In~other words,
it satisfies the product rule:
\begin{equation}
\begin{split}
\nabla(\gamma(w)\varphi)=\gamma(\nabla^gw)\varphi+\gamma(w)\nabla\varphi,&\quad\forall\varphi,\psi\in C^\infty(M,W),  \\
 &\quad\forall w\in C^\infty(M, {\rm Cl}(M,g)).
\end{split}
\end{equation}
\end{enumerate}
\end{defi}

\begin{ex}{\bf (Exterior algebra bundle as a Dirac Bundle).}\label{forms}
Let $(M,g)$ be a $C^\infty$ Riemannian manifold of dimension $m$. The~tangent and the cotangent bundles are identified by the $\flat$-map defined by $v^{\flat}(w):=g(v,w)$. Its inverse is denoted as $\sharp$. The~exterior algebra can be seen as a Dirac bundle after the following~choices:
\begin{itemize}
  \item $W:=\Lambda(T^*M)=\bigoplus_{j=0}^m \Lambda^j(T^*M)$: exterior algebra over $M$;
  \item $\langle \cdot,\cdot \rangle $: Riemannian structure induced by $g$;
  \item $\nabla$: lift of the Levi Civita connection;
  \item By means of interior and exterior multiplication, $\text{int}(v)\varphi:=\varphi(v,\cdot)$ and
  $\text{ext}(v)\varphi:=v^{\flat}\wedge\varphi$, we can define the following:
\begin{equation}
    \begin{split}
        \gamma&:
            \begin{array}{lll}
                TM&\longrightarrow& Hom(W)\\
                v&\longmapsto&\gamma(v):=\text{ext}(v)-\text{int}(v).
            \end{array}
    \end{split}
  \end{equation}
  Recall that, since $\gamma^2(v)=-g(v,v){\bf 1}$, based on~the universal property, the~map
$\gamma$ extends uniquely to a real algebra bundle endomorphism $\gamma: {\rm Cl}(M,g)\longrightarrow {\rm Hom}(W)$.
\end{itemize}
\end{ex}

\begin{defi}
Let $(W,\left<\cdot,\cdot\right>,\nabla,\gamma)$ be a Dirac bundle over the
Riemannian manifold $(M,g)$. The~\textbf{Dirac operator}
$Q:C^\infty(M,W)\to C^\infty(M,W)$ is defined by
$Q:=\gamma\circ (\sharp\otimes\mathbb{1}) \circ\nabla$
\begin{equation}
    \begin{CD}
     {C^{\infty}(M,W)} @>{\nabla}>> {C^{\infty}(M,T^*M\otimes W)} \\
       @V{Q:=\gamma\circ(\sharp\otimes\mathbb{1})\circ\nabla}VV  @VV{\sharp\otimes\mathbb{1}}V \\
      {C^{\infty}(M,W)} @<{\gamma}<< {C^{\infty}(M,TM\otimes W)}
    \end{CD}
 \end{equation}

The square of the Dirac operator $P:=Q^2:C^\infty(M,W)\to
C^\infty(M,W)$ is called the \textbf{Dirac Laplacian}.
\end{defi}
\begin{defi}[\textbf{Dirac Complex}]\label{DC}
Let $Q$ be the Dirac operator for the Dirac bundle $W$
over the Riemannian manifold $(M,g)$ and $T\in {\rm Hom}(W)$. $(Q,T)$
is called a {\it Dirac complex} if and only if $T^2={\bf 1}$ and $QT=-TQ$.
We introduce the following notation:
\begin{equation}
\Pi_{\pm}:=\frac{{\bf 1} \mp T}{2} \qquad W_{\pm}:=\Pi_{\pm}(W)\qquad Q_{\pm}:=Q|_{C^{\infty}(M,W_{\pm})}.
\end{equation}
\end{defi}
\begin{remark}
The terminology introduced in Definition \ref{DC} is justified by the following~properties:
\begin{itemize}
\item $Q_{\pm}:C^\infty(M,W_{\pm})\longrightarrow C^\infty(M,W_{\mp})$;
\item $Q=\left[\begin{matrix}0&Q_{-}\\Q_{+}&0\end{matrix}\right]:C^\infty(M,\underbrace{W_{+}\oplus W_{-}}_{W})\longrightarrow C^\infty(M,\underbrace{W_{+}\oplus W_{-}}_{W})$;
\item the sequence
   \begin{equation*}
   \begin{CD}0@>>>{C^\infty(M,W_{+})}@>{Q_{+}}>>{C^\infty(M,W_{-})}@>{Q_{-}}>> {C^\infty(M,W_{+})} @>>> 0 \end{CD}
   \end{equation*}
 is a complex, i.e.,~$Q_{-}Q_{+}=0$.
\end{itemize}
\end{remark}

\begin{ex}{\bf (Exterior algebra bundle as a Dirac Bundle---Continuation).}\label{forms2}
The Dirac operator $Q=d+\delta$ is termed a Euler operator. We define the vector bundle isomorphism on the exterior algebra bundle $T$ as $T\eta:=(-1)^j\eta$ for $\eta\in\Lambda^j(T^*M)$ and extend it by linearity to $\Lambda(T^*M)$ in order to obtain the Dirac complex $(Q,T)$, termed the rolled-up De Rham complex.
\end{ex}

 The Dirac operator $Q$ is elliptic and symmetric as an operator on the Hilbert space $L^2(M,W)$, with the appropriate choice of domain of definition.
\begin{defi}[\bf{Analytical Index}] Let $(Q,T)$ be a Dirac complex over a compact manifold. If~$\partial M =\varnothing$, then the analytical index of the complex is defined as follows:
\begin{equation}
\text{Index}_a(Q,T):=\dim\ker(Q_{+})-\dim\ker(Q_{-}).
\end{equation}

If $\partial M \neq\varnothing$, and~there exists an elliptic boundary condition $B$ for which $Q_B$ is symmetric, then the analytical index of the complex with respect to this boundary condition is defined as follows:
\begin{equation}
\text{Index}_a(Q,T,B):=\dim\ker((Q_{+})_B)-\dim\ker((Q_{-})_B).
\end{equation}
\end{defi}
\begin{theorem}[\textbf{Atiyah--Patodi--Singer}] Let $(Q,T)$ be a Dirac complex over a compact manifold. If~ $\partial M =\varnothing$, then
\begin{equation}
\text{Index}_a(Q,T)=\text{Index}_t(Q,T),
\end{equation}
\noindent where $\text{Index}_t(Q,T)$ is a topological index, i.e.,~depending only on the topology of $M$ and $W$. If~ $\partial M \neq\varnothing$ and $B$ is an elliptic boundary condition, then
\begin{equation}\label{APS}
\text{Index}_a(Q,T,B)=\text{Index}_t(Q,T)+\textit{Boundary Term}(B).
\end{equation}
\end{theorem}
For generic Dirac bundles, the existence of local elliptic boundary conditions is not guaranteed. But~for the exterior algebra bundle, the absolute and relative boundary conditions are always local elliptic boundary conditions for the Euler operator. In~this case, the Atiyah--Singer index theorem takes the form of the Gauss--Bonnet--Chern theorem:
\begin{ex}{\bf (Exterior algebra bundle as a Dirac Bundle---Continuation).}\label{forms3}
In the boundaryless case we have (see~\cite{Gi95} page 179 and~\cite{BGV96} page 59)
\begin{equation}
(2\pi)^{\frac{m}{2}}\int_M\text{Pf}(-R^{\Lambda(T^*M)})=\chi(M),
\end{equation}
\noindent where $\text{Pf}(-R^{\Lambda(T^*M)})$, termed a Euler form, is the Pfaffian of the exterior algebra curvature and $\chi(M)$ is the Euler characteristic, which is a topological invariant of the manifold defined as:
\begin{equation}
\chi(M):=\sum_{j=0}^mb_j(M),
\end{equation}
\noindent where $b_j(M):=\dim H^j(M)$ is the $j$-th Betti number of $M$, the~dimension of the $j$-the De Rham absolute cohomology group, which is isomorphic to the $j$-th homology group. In~the boundaryless case, absolute and relative cohomology coincide, the~Hodge star operator defines an isomorphism between the $j$-th and the $m-j$-th De Rham cohomology, and hence,  the~Euler characteristic always vanishes if $m$ is~odd.

If the manifold $M$ has a boundary on which we impose the absolute (or relative) boundary condition, the~equality (\ref{APS}) becomes
\begin{equation}\label{GBC}
(2\pi)^{\frac{m}{2}}\int_M\text{Pf}(-R^{\Lambda(T^*M)})=\chi(M)+\int_{\partial M}\Phi(R^{\Lambda(T^*M)}, L(\partial M,M)))d\text{vol}_{\partial M},
\end{equation}
\noindent where $\Phi$ is a function of the curvature $R^{\Lambda(T^*M)}$ and of the second fundamental form $L$ of the embedding $\partial M\rightarrow M$ (see~\cite{Gi95} page 201). Note that in the case of a manifold with a boundary case, the Euler characteristic does not have to vanish if $m$ is odd.
\end{ex}
Next, we can apply the definitions above to recognize a Dirac bundle containing all the information required by the market model. In~the space of all possible strategies $M:=[0,T]\times\mathfrak{X}\subset\mathbb{R}^{N+1}$, now we introduce the Riemannian structure induced by the Euclidean metric in $\mathbb{R}^{N+1}$.
The cash flow bundle $\mathcal{V}$ has infinite rank and is therefore an unfavourable candidate for the Atiyah--Singer/Gauss--Bonnet--Chern theorem. Therefore, we choose a partition $t_0:=0<t_1<\dots<t_{n-1}<t_n:=T$ of the interval $[0,t]$ such that its mesh $\max_{1\le i \le n}(t_{i}-t_{i-1})\rightarrow 0\;(n\rightarrow+\infty)$ holds true. We repeat the construction of the market fibre bundle $\mathcal{B}_n$ in Definition \ref{MFB}, choosing as structure group $G_n:=\{\pi\in G\,|\,{\rm supp}(\pi)\subset\{t_0,\dots,t_n\}\}$ with Lie algebra $\mathfrak{g}_n=\mathbb{R}^{n+1}$, and~the construction of the cash flow bundle ${\mathcal V}_n$ in Section \ref{cfbundle}, choosing as fibre $V_n:=\mathbb{R}^{n+1}$. Note that in this case, by~standard elliptic theory, the~eigenspaces of the connection Laplacian under the Neumann boundary condition are all finite dimensional.  In~particular, there are finitely many linear, independent  Radon--Nikodym derivatives that can perform a change of measure from the statistical to a risk-neutral~one.\par
The exterior algebra bundle over $M$ twisted with the cash flow bundle ${\mathcal V}_n$ is given the structure of a Dirac bundle with the following~choices:
\begin{itemize}
\item $g$: restriction of the Euclidean metric;
  \item ${W_n}:=\Lambda(T^*M)\otimes\mathcal{V}_n$: twisted bundle of finite rank $(n+1)2^{N+1}$;
  \item $\langle \eta_1\otimes v_1,\eta_2\otimes v_2\rangle^{W_n}:=\langle \eta_1,\eta_2\rangle ^{\Lambda(T^*M)}\langle v_1,v_2\rangle^{\mathcal{V}_n}$: Riemannian structure;
  \item $\nabla^{W_n}:=\nabla^{\Lambda(T^*M)}\otimes {\bf 1}_{{\mathcal V}_n}+{\bf 1}_{\Lambda(T^*M)}\otimes\nabla^{{\mathcal V}_n}$: connection;
  \item $\gamma^{{W_n}}:=\gamma^{\Lambda(T^*M)}\otimes {\bf 1}_{{\mathcal V}_n}$: real algebra bundle endomorphism $\gamma:{\rm Cl}(M,g)\longrightarrow {\rm Hom}({W_n})$;
  \item $T^{{W_n}}:=T^{\Lambda(T^*M)}\otimes {\bf 1}_{{\mathcal V}_n}\in {\rm Hom}({W_n})$: a symmetry anticommuting with the Dirac operator.
\end{itemize}

 Properties (iv)--(vi) of Definition \ref{DiracBundle} are satisfied, as highlighted on page 226 of~\cite{Gi95}, and~$T^{{W_n}}$ defines via Definition \ref{DC} a Dirac complex. Since for the curvature we have the following:
\begin{equation}
R^{{W_n}}=R^{\Lambda(T^*M)}\otimes {\bf 1}_{{\mathcal V}_n}+{\bf 1}_{\Lambda(T^*M)}\otimes R^{{\mathcal V}_n}={\bf 1}_{\Lambda(T^*M)}\otimes R^{{\mathcal V}_n},
\end{equation}
because $M$ carries the flat Euclidean metric ($R^{\Lambda(T^*M)}=0$), the~Gauss--Bonnet--Chern theorem reads as follows:
\begin{align}
&(2\pi)^{\frac{N+1}{2}}\int_M\text{Pf}(-{\bf 1}_{\Lambda(T^*M)}\otimes R^{{\mathcal V}_n}) \nonumber \\
& \hspace{1cm} =\text{Rank}({\mathcal V}_n)\left[\chi(M)+\int_{\partial M}\Phi(R^{\Lambda(T^*M)}, L(\partial M,M)))d\text{vol}_{\partial M}\right].
\end{align}\label{Obs0}

The integrand of the boundary term vanishes because the Riemannian curvature is zero, as~one can see in~\cite{Gi95} page 199. Therefore,
\begin{equation}\label{Obs1}
(2\pi)^{\frac{N+1}{2}}\int_M\text{Pf}(-{\bf 1}_{\Lambda(T^*M)}\otimes R^{{\mathcal V}_n})=\text{Rank}({\mathcal V}_n)\chi(M).
\end{equation}

Utilizing the fact that
\begin{equation}
\chi(M)=\chi([0,T]\times\mathfrak{X})=\chi(\mathfrak{X}),
\end{equation}
and inserting the value for the cash flow bundle rank, we obtain for all $n\in\mathbb{N}_1$
\begin{equation}\label{Obs2}
\boxed{
\frac{1}{n+1}\left(\frac{\pi}{2}\right)^{\frac{N+1}{2}}\int_M\text{Pf}(-{\bf 1}_{\Lambda(T^*M)}\otimes R^{{\mathcal V}_n})=\chi(\mathfrak{X}),
}
\end{equation}
from which we see that a non-vanishing Euler characteristic of the space of all possible nominals is a topological obstruction for the market model to be arbitrage-free. To~summarize:
\begin{theorem}
If a market model has a bounded space of asset nominals $\mathfrak{X}$, whose Euler characteristic $\chi(\mathfrak{X})$ does not vanish, then the zero curvature condition (ZC), or equivalently the (NUPBR) condition, and a~fortiori the~NFLVR condition, cannot be satisfied:
\begin{equation}\label{char}
\text{(NFLVR)}\Rightarrow\text{(NUPBR)}\Rightarrow\chi(\mathfrak{X})=0.
\end{equation}
\end{theorem}

\begin{remark}
Formula (\ref{Obs2}) says that the total quantity of arbitrage allowed by a market over all asset strategies is a topological invariant of the asset nominal space.
\end{remark}

\begin{remark}
The Euler characteristic of the nominal space $\chi(\mathfrak{X})$ as topological obstruction to NFLVR is consistent with the results in~\cite{Fa15}, where the first homotopy group $\Pi_1(\mathfrak{X},D)$  must be trivial if NFLVR holds true.
\end{remark}

\begin{remark} There is a more general version of the index theorem which holds true for any elliptic operator on a vector bundle over a manifold with boundary that has an elliptic boundary condition, which is  a generalization of the Atiyah--Singer--Patodi index theorem. For~the connection Laplacian $P:={\nabla^{{\mathcal V}}}^*\nabla^{{\mathcal V}}$ over the cash flow bundle ${\mathcal V}$ under the Neumann boundary condition $B_N$, it has the following form:
\begin{equation}\label{indexP}
\text{Index}_a(P,B_N)=\text{Index}_t(P)+\textit{Boundary Term}(P,B_N).
\end{equation}

 Do we expect from (\ref{indexP}) another topological obstruction to arbitrage than the one induced by the Gauss--Bonnet--Chern theorem? We want to answer this question by means of qualitative reasoning. We observe~that:
\begin{itemize}
\item $\text{Index}_a(P,B_N)$ is an integral over $M$ of an expression depending on the derivatives of Christoffel's symbols for the connection $\nabla^{{\mathcal V}}$;
\item $\text{Index}_t(P,B_N)$ is a topological invariant of the manifold $M$;
\item $\textit{Boundary Term}(P,B_N)=0$, because~of the choice of the boundary condition and Green's formula.
\end{itemize}

 Therefore, the~information contained in (\ref{indexP}) is essentially the same as that in (\ref{GBC}), and thus, we expect no additional insights.
\end{remark}

\subsection{Topological Obstruction to Arbitrage Induced by the Bochner--Weitzenb\"ock~Theorem}
Let us now introduce two local elliptic boundary condition which guarantee self-adjointness  to the Dirac operator and the Dirac Laplacian on $\Lambda(T^*M)$:

\begin{itemize}
\item \textbf{Absolute~boundary~condition:} 
\begin{equation}
\begin{split}
B^0_{\text{abs}}(f)&:=(\text{int}(\nu)(f))|_{\partial M}\\
B^1_{\text{abs}}(f)&:=B^0_{\text{abs}}(f)\oplus B^0_{\text{abs}}(Q^{\Lambda(T^*M)}f),
\end{split}
\end{equation}
where the operation int is the interior multiplication in $\Lambda(T^*M)$, defined as
\begin{equation}
\text{int}(v)(f):=f(v,\cdot).
\end{equation}
\item \textbf{Relative boundary condition:}
\begin{equation}
\begin{split}
B^0_{\text{rel}}(f)&:=(\text{ext}(\nu)(f))|_{\partial M}\\
B^1_{\text{rel}}(f)&:=B^0_{\text{rel}}(f)\oplus B^0_{\text{rel}}(Q^{\Lambda(T^*M)}f),
\end{split}
\end{equation}
where the operation ext is the exterior multiplication in $\Lambda(T^*M)$, defined as
\begin{equation}
\text{ext}(v)(f):=v^{\#}\wedge  f,
\end{equation}
\noindent and $\#$ is the isomorphism between $TM$ and $T^*M$.
\end{itemize}

\begin{theorem}\label{equiv_NFLVR_2}
The following equations hold:
\begin{equation}\label{BW}
\boxed{
(Q^W_{B^0_{\text{abs}}\otimes B_N})^2=({\nabla^W}^*\nabla^W)_{B^1_{\text{abs}}\otimes B_N}.
}
\end{equation}

 Moreover, the~no-free-lunch-with-vanishing-risk condition is satisfied if and only if the Dirac Laplacian admits harmonic sections of the cash flow bundle:
 \vspace{12pt}
\begin{equation}\label{equiv_NFLVR}
\boxed{
\text{(NFLVR)}\Longleftrightarrow 0\in {\rm spec}_d\left((Q^W)^2_{B^1_{\text{abs}}\otimes B_N}\right).
}
\end{equation}
\end{theorem} 
\begin{proof}
The first equation in (\ref{BW}) follows from the Bochner--Weitzenb\"ock formula, which holds true for any Dirac bundle (cf.~\cite{LaMi89} page 155):
\begin{equation}\label{BWT}
(Q^W_{B^0_{\text{abs}}\otimes B_N})^2=({\nabla^W}^*\nabla^W)_{B^1_{\text{abs}}\otimes B_N}+\mathcal{R}^W,
\end{equation}
\noindent where
\begin{equation}
\mathcal{R}^W:=\sum_{\mu,\nu=0}^N\gamma^W(e_\mu)\gamma^W(e_\nu)R^W(e_\mu,e_\nu)
\end{equation}
\noindent is the curvature homomorphism on the vector bundle $W$ and is independent of the choice of an o.n. frame of $TM$. Inserting
\begin{equation}
R^W={\bf 1}^{\Lambda(T^*M)}\otimes R^{\mathcal{V}}
\end{equation}
\noindent into Equation~(\ref{BWT}), and decomposing a section $\psi$ of the twisted vector bundle as
\begin{equation}
\psi=\sum_{i,j}a_{i,j}c_i\otimes f_j,
\end{equation}
\noindent where $\{f_j\}_j$ an o.n.b of $L^2(\Omega, \mathcal{V}, \mathcal{A}_0,dP)$, and $\{c_i\}_i$ is an o.n.b of $L^2(\Omega, \Lambda(T^*M), \mathcal{A}_0,dP)$
of eigenvectors of $\gamma^{\Lambda(T^*M)}(e_{\mu})$ such that for~all $\mu$,
\begin{equation}
\gamma^{\Lambda(T^*M)}(e_{\mu})c_i=(-1)^i\imath,
\end{equation}
\noindent we obtain
\begin{equation}
\begin{split}
(\mathcal{R}^W\psi,\psi)&=-\sum_{\substack{i,j,k,l \\ \mu\neq\nu}}a_{i,j}a_{k,l}(-1)^{\mu+\nu}\delta_{i,k}(R^{\mathcal{V}}(e_\mu,e_\nu)f_j,f_l)\\
&=\sum_{\substack{i,j,k,l \\ \nu\neq\mu}}a_{i,j}a_{k,l}(-1)^{\nu+\mu}\delta_{i,k}(R^{\mathcal{V}}(e_\nu,e_\mu)f_j,f_l)\\
&=-(\mathcal{R}^W\psi,\psi),
\end{split}
\end{equation}
\noindent and hence, for~any $\psi$
\begin{equation}
(\mathcal{R}^W\psi,\psi)=0.
\end{equation}

Based on the polarization identity, we infer $\mathcal{R}^W=0$, and Equation~(\ref{BW}) follows.
For the proof of (\ref{equiv_NFLVR}), we~have:
\begin{itemize}
\item[$\Rightarrow:$] If NFLVR holds, then with Theorem \ref{spec_NFLVR} there exists an $f\in\text{dom}\left(\Delta^{\mathcal{V}}_{B_N}\right)$ such that $\Delta^{\mathcal{V}} f=0$, and~$\nabla^{\mathcal{V}}=0$. Let us consider the section $\psi:=c\otimes f$, where $c\in\mathbb{R}$ is a constant. With~Equation~(\ref{BW}), we infer via integration by parts
\begin{equation}
     ({Q^W}^2\psi,\psi)=({\nabla^W}^*\nabla^W\psi,\psi)=(\nabla^W\psi,\nabla^W\psi) = c^2(\nabla^{\mathcal{V}}f,\nabla^{\mathcal{V}}f)=0,
    \end{equation}
    \noindent and the proof in this direction is complete.
\item[$\Leftarrow:$] If there is a $\psi\neq0$ in $\ker\left(Q^W_{B^0_{\text{abs}}\otimes B_N})^2\right)$, then
\begin{equation}
    \psi=\bigoplus_{i,j}a_{i,j}c_i\otimes f_j,
    \end{equation}
    where $c_i\in\Lambda^i(T^*M)$, because~the Laplace--Beltrami operator maintains the degree of differential forms.
    Therefore, there is at least a pair $(i,j)$ such that for $c_i\neq0$ and $f_j\neq0$
\begin{equation}
    {Q^W}^2c_i\otimes f_j=0,
    \end{equation}
    \noindent and based on~Equation~(\ref{BW}),
\begin{equation}
    \begin{split}
    0&=({Q^W}^2c_i\otimes f_j,c_i\otimes f_j)=(\nabla^W(c_i\otimes f_j),\nabla^W(c_i\otimes f_j))\\
    &=(\nabla^{\Lambda^i(T^*M)}c_i,\nabla^{\Lambda^i(T^*M)}c_i)(\nabla^{\mathcal{V}}f_j,\nabla^{\mathcal{V}}f_j),
    \end{split}
    \end{equation}
    \noindent and therefore,  $\nabla^{\mathcal{V}}f_j=0$ for at least a $j$. Hence, $\Delta^{\mathcal{V}}f_j=0$, and~NFLVR follows from Theorem \ref{spec_NFLVR}.
\end{itemize}\end{proof}

\begin{remark}
Theorem \ref{equiv_NFLVR_2} holds true if we replace the absolute with the relative boundary condition.
\end{remark}

\begin{corollary}\label{homology}
The no-free-lunch-with-vanishing-risk condition is satisfied if and only if $H_*(M,{\mathcal V})$, the~homology group of the vector bundle $V$ over $M$, does not vanish:
\begin{equation}
\boxed{
\text{(NFLVR)}\Longleftrightarrow H_*(M,{\mathcal V})\neq\{0\}.
}
\end{equation}
\end{corollary}
\begin{proof}
According to Theorem \ref{equiv_NFLVR_2}, NFLVR and the isomorphism between the cohomology group and the kernel of $(Q^W_{B^0_{\text{abs}}\otimes B_N})^2$
\begingroup\makeatletter\def\f@size{9.5}\check@mathfonts
\def\maketag@@@#1{\hbox{\m@th\normalsize\normalfont#1}}%
\begin{equation}
  H^*(M,\Lambda(T^*M)\otimes {\mathcal V})=\bigoplus_{j=0}^NH(M,\Lambda^j(T^*M)\otimes H^*(M,{\mathcal V})=\bigoplus_{j=0}^NH^j(M)\otimes H^*(M,{\mathcal V}),
\end{equation}
\endgroup
\noindent where $H^j(M,\mathbb{R})$ is the $j$th cohomology group of $M$ and $H^*(M,{\mathcal V})$ the cohomology group of the vector bundle $V$ over $M$.
We see that NFLVR  is satisfied if and only if there exists at least a $j=0,\dots,N+1$ such that
\begin{equation}
H^j(M,\mathbb{R})\otimes H^*(M,{\mathcal V}),
\end{equation}

 Since $H^0(M)\neq\{0\}$ because the constant functions are eigenvectors of the Laplace--Beltrami operator under the absolute boundary condition for the eigenvalue $0$, this condition is satisfied if and only if
the cohomology group of the cash flow bundle does not vanish. Since the cohomology group is isomorphic to the homology group, the~proof is~complete.
\end{proof}

\begin{remark}
Corollary \ref{homology} states that the vanishing of the homology group of the cash flow bundle is a topological obstruction to the no-free-lunch-with-vanishing-risk condition.
\end{remark}

\section{Conclusions}\label{section5}
By introducing an appropriate stochastic differential geometric
formalism, the~classical theory of stochastic finance can be
embedded into a conceptual framework called Geometric
Arbitrage Theory, where the market is modelled with a principal
fibre bundle with a connection and arbitrage corresponding to its curvature. The associated vector bundle, termed cash flow bundle, carries a covariant differentiation induced by the connection.
The presence of the eigenvalue $0$ in the spectrum of the connection Laplacian on the cash flow bundle, or of the Dirac Laplacian on the cash flow bundle twisted with the exterior algebra bundle,
characterizes the fulfilment of the no-free-lunch-with-vanishing-risk condition for the market model. We extend the Jarrow--Protter--Shimbo bubble theory to markets allowing for arbitrage and highlight the connections with the Platen--Heath real world pricing of the benchmark approach.
The non-vanishing of the Euler characteristic of the asset nominal space and the non-vanishing of the homology group of the cash flow bundle are \textit{topological} obstructions to the fulfillment of the NFLVR
condition.\par As a result, we have justified the title of this paper, which is an adaptation of Ka\v{c}'s famous question that introduced spectral inverse problems for self-adjoint operators on manifolds: \textit{Can you hear the shape of a drum?} We do hope that this approach will be appreciated by the mathematical finance community. 

\appendix
\setcounter{section}{0}
\section{Generalized Derivatives of Stochastic~Processes}\label{Derivatives}
In stochastic differential geometry, one would like to lift
the constructions of stochastic analysis from open subsets of
$\mathbf{R}^N$ to  $N$ dimensional differentiable manifolds. To~that
aim, chart invariant definitions are needed, and hence, a stochastic
calculus satisfying the usual chain rule instead of It\^{o}'s Lemma is required,
(cf.~\cite{HaTh94}, Chapter 7, and~the remark in Chapter 4 at the
beginning of page 200). That is why the papers about geometric arbitrage theory are mainly concerned
 with stochastic integrals and derivatives meant in \textit{Stratonovich}'s
sense and not in \textit{It\^{o}}'s. Of~course, at~the end of the computation, Stratonovich integrals can be transformed into It\^{o}'s.
Note that a fundamental portfolio equation, the~self-financing condition, cannot be directly formally expressed with Stratonovich integrals; it must~first be expressed with It\^{o}'s and then transformed into Stratonovich's because~it is a non-anticipative condition.
\begin{defi}\label{Nelson}
Let $I$ be a real interval and $Q=(Q_t)_{t\in I}$ be a  $\mathbf{R}^N$-valued stochastic process on the probability space
$(\Omega, \mathcal{A}, P)$. The~process $Q$ determines three families of $\sigma$-subalgebras of the $\sigma$-algebra $\mathcal{A}$:
\begin{itemize}
\item[(i)] ``Past'' $\mathcal{P}_t$, generated by the preimages of Borel sets in $\mathbf{R}^N$  by all mappings $Q_s:\Omega\rightarrow\mathbf{R}^N$ for $0<s<t$;
\item[(ii)] ``Future'' $\mathcal{F}_t$, generated by the preimages of Borel sets in $\mathbf{R}^N$  by all mappings $Q_s:\Omega\rightarrow\mathbf{R}^N$ for $0<t<s$;
\item[(iii)] ``Present'' $\mathcal{N}_t$, generated by the preimages of Borel sets in $\mathbf{R}^N$  by the mapping $Q_s:\Omega\rightarrow\mathbf{R}^N$.
\end{itemize}

Let $Q=(Q_t)_{t\in I}$ be continuous.
 Assuming that the following limits exist,
\textbf{Nelson's stochastic derivatives} are defined as follows:
\begin{equation}
\boxed{
\begin{split}
&\mathfrak{D}Q_t:=\lim_{h\rightarrow
0^+}\mathbb{E}\Br{\left.\frac{Q_{t+h}-Q_t}{h}\right|
\mathcal{P}_t}\text{: forward derivative,}\\
& \mathfrak{D}_*Q_t:=\lim_{h\rightarrow
0^+}\mathbb{E}\Br{\left.\frac{Q_{t}-Q_{t-h}}{h}\right|
\mathcal{F}_{t}}\text{: backward derivative,}\\
&\mathcal{D}Q_t:=\frac{\mathfrak{D}Q_t+\mathfrak{D}_*Q_t}{2}\text{: mean derivative}.
\end{split}
}
\end{equation}

Let $\mathcal{S}^1(I)$ be the set of all processes $Q$ such that
$t\mapsto Q_t$, $t\mapsto \mathfrak{D}Q_t$ and $t\mapsto \mathfrak{D}_*Q_t$ are continuous
mappings from $I$ to $L^2(\Omega, \mathcal{A})$. Let
$\mathcal{C}^1(I)$ be the completion of $\mathcal{S}^1(I)$ with respect
to the norm
\begin{equation}
\boxed{
\|Q\|:=\sup_{t\in I}\br{\|Q_t\|_{L^2(\Omega, \mathcal{A})}+\|\mathfrak{D}Q_t\|_{L^2(\Omega, \mathcal{A})}+\|\mathfrak{D}_*Q_t\|_{L^2(\Omega, \mathcal{A})}}.
}
\end{equation}
\end{defi}

\begin{remark}
The stochastic derivatives $\mathfrak{D}$, $\mathfrak{D}_*$, and  $\mathcal{D}$
correspond to It\^{o}'s, to~the anticipative, and, respectively,  to~Stratonovich's integral (cf.~\cite{Gl11}). The~process space $\mathcal{C}^1(I)$ contains all of It\^{o} processes. If~$Q$ is a Markov process, then the sigma algebras $\mathcal{P}_t$ (``past'') and $\mathcal{F}_t$ (``future'') in the definitions of forward and backward derivatives can be substituted by the sigma algebra $\mathcal{N}_t$ (``present'') (see Chapter 6.1 and 8.1 in \cite{Gl11}).
\end{remark}

  Stochastic derivatives can be defined pointwise in $\omega\in\Omega$ outside the class $\mathcal{C}^1$ in terms of generalized functions.
\begin{defi}
Let $Q:I\times\Omega\rightarrow\mathbf{R}^N$ be a continuous linear functional in the test processes $\varphi:I\times\Omega\rightarrow\mathbf{R}^N$ for $\varphi(\cdot,\omega)\in C^{\infty}_c(I,\mathbb{R}^N)$. This means that for a fixed $\omega\in\Omega$, the functional $Q(\cdot,\omega)\in\mathcal{D}(I,\mathbb{R}^N)$, the~topological vector space of continuous distributions. We can then define
\textbf{Nelson's generalized stochastic derivatives:}
\begin{equation}
\boxed{
\begin{split}
&\mathfrak{D}Q(\varphi_t):=-Q(\mathfrak{D}\varphi_t)\text{: forward generalized derivative;}\\
& \mathfrak{D}_*Q(\varphi_t):=-Q(\mathfrak{D}_*\varphi_t)\text{: backward generalized derivative;}\\
&\mathcal{D}(\varphi_t):=-Q(\mathcal{D}\varphi_t)\text{: mean generalized derivative}.
\end{split}
}
\end{equation}
\end{defi}
If the generalized derivative is regular, then the process has a derivative in the classic sense. This construction is nothing other than a straightforward pathwise lift of the theory of generalized functions to a wider class of stochastic processes which do not a priori allow for Nelson's derivatives in the strong sense. We will utilize this feature in the treatment of credit risk, where many processes with jumps~occur.

\end{document}